\documentclass[pdflatex,sn-mathphys-num]{sn-jnl}

\usepackage{graphicx}%
\usepackage{multirow}%
\usepackage{amsmath,amssymb,amsfonts}%
\usepackage{amsthm}%
\usepackage{mathrsfs}%
\usepackage[title]{appendix}%
\usepackage{xcolor}%
\usepackage{textcomp}%
\usepackage{manyfoot}%
\usepackage{booktabs}%
\usepackage{algorithm}%
\usepackage{algorithmicx}%
\usepackage{algpseudocode}%
\usepackage{listings}%
\usepackage{subcaption}

\theoremstyle{thmstyleone}%
\newtheorem{theorem}{Theorem}
\newtheorem{proposition}[theorem]{Proposition}%

\theoremstyle{thmstyletwo}%
\newtheorem{lemma}{Lemma}%
\newtheorem{corollary}{Corollary}

\theoremstyle{thmstylethree}%

\newcommand{\be}{\begin{eqnarray}}
\newcommand{\ee}{\end{eqnarray}}
\newcommand{\bea}{\begin{eqnarray*}}
\newcommand{\eea}{\end{eqnarray*}}

\newcommand{\pa}{\pmb{a}}
\newcommand{\px}{\pmb{x}}
\newcommand{\py}{\pmb{y}}
\newcommand{\pe}{\pmb{e}}

\newcommand{\pv}{\pmb{v}}
\newcommand{\pz}{\pmb{z}}

\newcommand{\pI}{\pmb{I}}

\newcommand{\pA}{\pmb{A}}

\newcommand{\pU}{\pmb{U}}
\newcommand{\pV}{\pmb{V}}
\newcommand{\pZ}{\pmb{Z}}

\newcommand{\pX}{\pmb{X}}

\newcommand{\pE}{\pmb{E}}
\newcommand{\pP}{\pmb{P}}

\newcommand{\hx}{\hat{\pmb{x}}}

\newcommand{\pSig}{\pmb{\Sigma}}

\newcommand{\mbC}{\mathbb{C}}
\newcommand{\mbE}{\mathbb{E}}

\newcommand{\mri}{\mathrm{i}}
\newcommand{\mrd}{\mathrm{d}}

\newcommand{\tp}{\texttt{T}}

\DeclareMathOperator*{\rank}{rank}
\DeclareMathOperator*{\tr}{tr}
\DeclareMathOperator*{\Unif}{Unif}
\DeclareMathOperator*{\dist}{dist}

\begin{document}

\title[Non-asymptotic Analysis of Expected Reconstruction Risk]{Non-asymptotic Analysis of Expected Reconstruction Risk for Trigonometric Polynomial Models}


\author[1]{\fnm{Hang} \sur{Xu}}\email{hangxu@zstu.edu.cn}

\author*[2]{\fnm{Yi} \sur{Shen}}\email{yshen@zstu.edu.cn}

\affil[1,2]{\orgdiv{Department of Mathematics}, \orgname{Zhejiang Sci-Tech University}, \orgaddress{\city{Hangzhou}, \postcode{310018}, \country{P. R. China}}}



\abstract{We investigate the expected reconstruction risk of trigonometric polynomial models under different sampling schemes.
Through numerical experiments, we observe that when the sampling nodes $\{t_l\}_{l=1}^m$ are i.i.d. random variables uniformly distributed over $[0,1)$, the associated structured random matrix $\pA \in \mbC^{m \times N}$ with $A_{l,k} = e^{2\pi \mathrm{i} kt_l}, k \in \Gamma = \{-q, \dots, q\}, N = 2q+1$ frequently becomes nearly singular or severely ill-conditioned. As a consequence, the expected reconstruction risk exhibits divergent behavior. 
In contrast, when the sampling nodes $t_l$ are either equidistant points or small random perturbations of an equidistant grid, the expected reconstruction risk undergoes a sharp phase transition at the interpolation threshold $m=N$.
To better understand the underlying mechanisms behind these different phenomena, we characterize the expected reconstruction risk through the spectral quantity $\sum_{i=1}^{r} \frac{1}{\sigma_i^2(\pA)}$, where $\sigma_i(\pA)$ denotes the singular values of the sampling matrix. Based on this spectral representation, we theoretically prove that the expected reconstruction risk diverges under uniformly distributed random sampling. Furthermore, we derive an explicit formula for the expected reconstruction risk in the equidistant sampling case and establish upper and lower bounds for the expected reconstruction risk under jittered sampling.}

\keywords{Trigonometric Polynomials, Phase Transition, Non-asymptotic Analysis, Expected Reconstruction Risk}

\maketitle

\section{Introduction}
In this paper, we focus on multivariate trigonometric polynomials, which offer a particularly convenient and mathematically rich framework. These models arise naturally in scenarios where uniform sampling is unavailable, with applications spanning cardiology \cite{strohmer1996recover}, computed tomography \cite{averbuch2001fast}, nonuniform Fourier analysis \cite{dutt1993fast}, geophysics \cite{rauth1998smooth}, and image processing \cite{strohmer1997computationally}. Beyond their practical utility, trigonometric polynomials provide a mathematically tractable setting for investigating how the singular value spectrum governs statistical recoverability. Specifically, the singular value distributions of trigonometric sensing matrices exhibit rich behaviors that depend on the sampling scheme, frequency configuration, and oversampling ratio. This makes them ideal candidates for studying interpolation phenomena \cite{muthukumar2020harmless}, double descent \cite{mei2022generalization}, and spectral transitions \cite{moitra2015super}.

\subsection{Sampling Models}
For a fixed constant $q$, we set 
\[
\Gamma = \{-q, \dots, q\}, \; N = 2q+1.
\]
We aim to recover the coefficient vector $\px \in \mathbb{C}^{N}$ of a random trigonometric polynomial defined by
\begin{align}\label{x}
	f(t) = \sum_{k \in \Gamma} x_k \, e^{2\pi \mri k t}, 
\end{align}
from noisy pointwise samples acquired at random locations. 
We consider two random models for the sampling nodes, as proposed in \cite{bass2005random}:
\begin{itemize}
\item[(a)] The $t_1,\dots , t_m$ are i.i.d. random variables, each of which is uniformly distributed over $[0, 1]$. 
\item[(b)] The jittered sampling nodes are obtained as random perturbations of a uniform grid, namely
\[
t_l=\frac{l-1}{m}+\delta_l,\qquad l=1,\dots,m,
\]
where $\delta_1,\dots,\delta_m$ are i.i.d. random variables uniformly distributed on $(0,T)$.
\end{itemize}
The entries of the associated structured random matrix $\pA$ are given by
\begin{equation}\label{eq:alk}
	A_{l,k} = e^{2\pi \mathrm{i} kt_l}, \quad l=1,\dots,m,\quad k\in\Gamma.
\end{equation}
Then the recovery problem is formulated as the linear system 
\begin{align} \label{continuous_model}
\py = \pA\px + \pe,
\end{align}
where   $\pe \in \mbC^m$ represents additive noise with  $\text{Re}(\pe), \text{Im}(\pe) \sim \mathcal{N}(0,\frac{\epsilon^2}{2m}\pI_m)$.
The regression coefficients $\hx$  is fitted with
\begin{equation}\label{eq:Adagger}
	\hx = \pA^\dagger \py.
\end{equation}
We estimate the risk of $\hx$ under a random model for $\px$ where
\begin{align} \label{xx}
\mbE[\px\px^*] = \pI_N,
\end{align}
which implies that $\mbE\|\px\|^2=N$.

Since the rank $r$ of matrix $\pA$ satisfies  $r \leq \min\{m,N\}$, we assume that the singular value decomposition (SVD) of $\pA$ is $\pA = \pU\pSig\pV^*$ with $\pU \in \mbC^{m \times r}$, $\pV \in \mbC^{N \times r}$ and $\pSig \in \mbC^{r \times r}$.
Thus
\begin{align}\label{A_dagger}
\pA^\dagger = \pV\pSig^{-1}\pU^*.
\end{align}
If $\pA$ is full row rank, then 
\begin{align}
\pA^\dagger = \pA^*(\pA\pA^*)^{-1}.
\end{align}
If $\pA$ is full column rank, then 
\begin{align}
\pA^\dagger = (\pA^*\pA)^{-1}\pA^*.
\end{align}
Generally, we have
\begin{align}
\hx = \pV\pV^*\px+\pV\pSig^{-1}\pU^*\pe.
\end{align}
Specifically, when $\pA$ is full column rank, then $\pV\pV^*=\pI_N$.  For how to compute \eqref{eq:Adagger} numerically in a fast way, please refer to \cite{dutt1993fast,potts2001fast} for the non-equidistant fast Fourier transform.

\subsection{Problem Set-up}
Despite this substantial body of literature, most existing analyses are probabilistic. Typical results provide high-probability bounds for the smallest singular value, the condition number, or reconstruction errors. Comparatively little is known about the expectation of the reconstruction error itself. This issue is particularly important because the expectation is much more sensitive to rare ill-conditioned events than high-probability estimates. Consequently, conclusions based solely on probabilistic stability do not necessarily imply finite expected reconstruction errors.
Motivated by this observation, we numerically investigate the expected reconstruction error under the above two sampling models. Our experiments reveal a striking contrast.

For independent uniform sampling, the expected reconstruction error grows rapidly as the number of Monte Carlo samples increases. Moreover, MATLAB frequently reports the warning:
\begin{quote}
``Matrix is close to singular or badly scaled. Results may be inaccurate.''
\end{quote}
This phenomenon indicates that nearly singular sampling matrices  occur sufficiently often to dominate the expectation. The numerical evidence strongly suggests that the expected error may diverge, despite the fact that stable reconstruction holds with overwhelming probability.

\begin{figure}[!t] 
\setlength{\abovecaptionskip}{-0.1cm}   
\setlength{\belowcaptionskip}{0.1cm}   
\centering 
\begin{minipage}[b]{0.48\textwidth} 
\centering 
\includegraphics[width=1\textwidth]{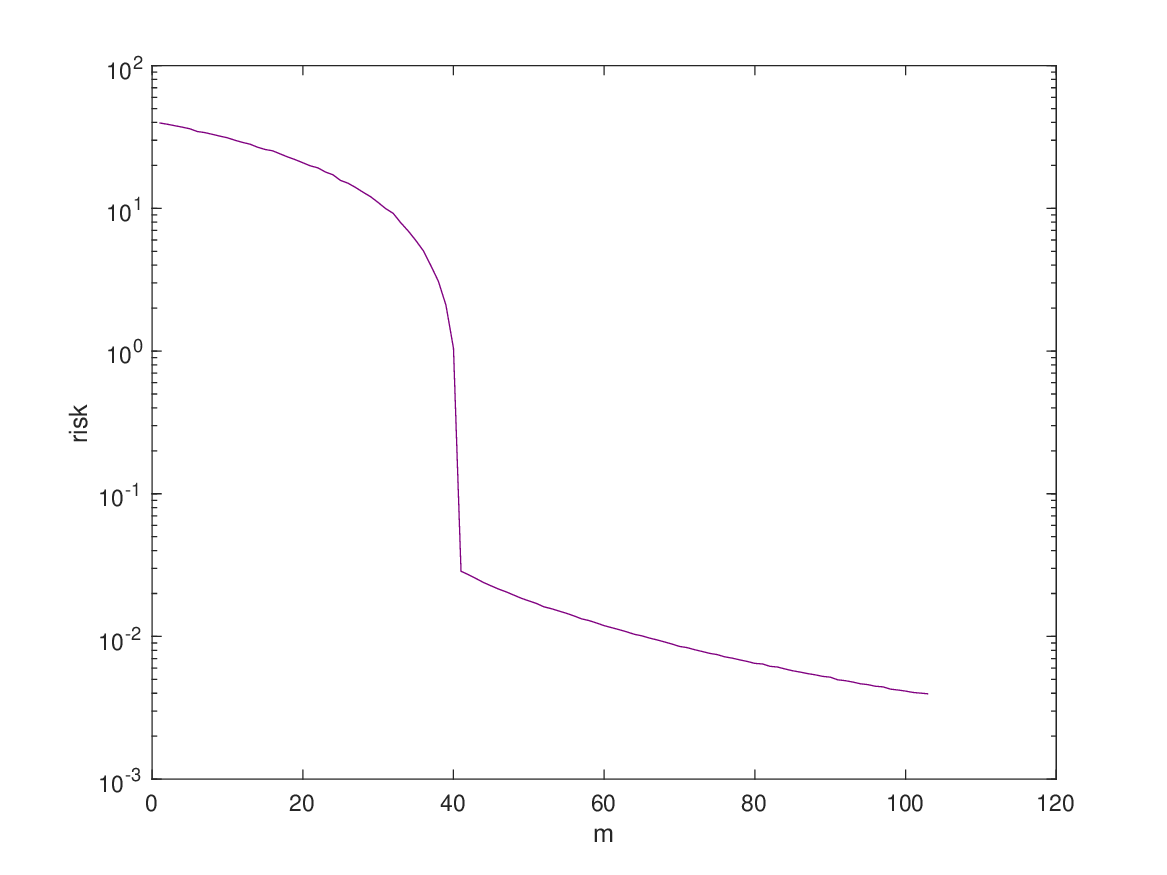} 
\vspace{-0.7cm}
\subcaption{$N=41$.}
\end{minipage}
\begin{minipage}[b]{0.48\textwidth} 
\centering 
\includegraphics[width=1\textwidth]{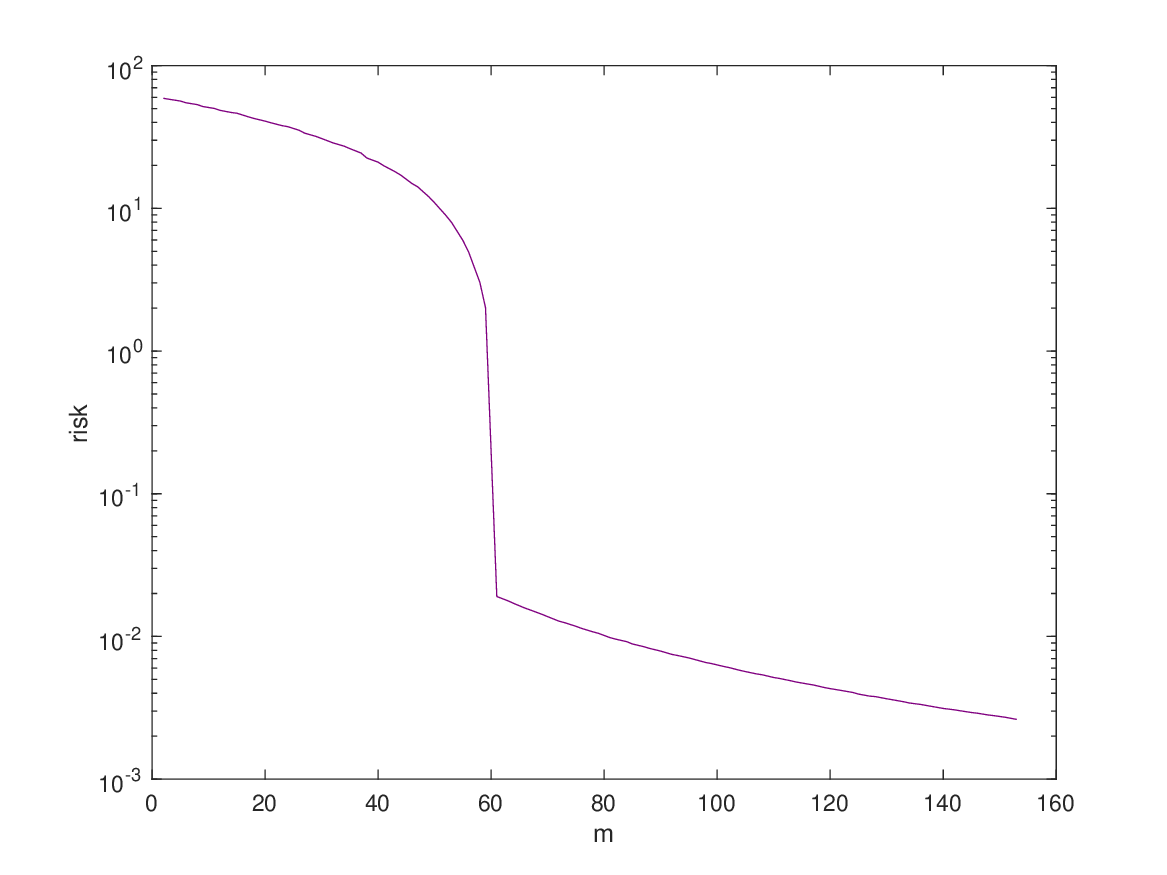} 
\vspace{-0.7cm}
\subcaption{$N=61$.}
\end{minipage}
  \caption{Plot of the average value of $\| \px - \hat{\px} \|^2$ as a function of $m$ under the random model for the jittered sampling nodes (b). Here, $N=41 \text{ or } 61$, each entry of the perturbation $\delta_l$ is drawn from i.i.d. $\Unif(0,\frac{1}{2m})$.
  A phase transition occurs at the critical sampling ratio $m/N=1$.}
\label{fig:random_risk}
\end{figure}

In contrast, jittered sampling exhibits a completely different behavior. As illustrated in Fig.~\ref{fig:random_risk}, the average error undergoes a sharp transition as the number of sampling points increases. Before the threshold $m = N$, the error decreases relatively slowly, whereas beyond this threshold it suddenly drops by several orders of magnitude and then continues to decay smoothly. 
Such a transition does not appear in independent uniform sampling, and its theoretical mechanism remains unclear.

Unlike deterministic uniform sampling, random sampling gives rise to a random Fourier--Vandermonde matrix whose spectral properties directly govern the statistical performance of least-squares reconstruction. In particular, the reconstruction risk can be expressed as a spectral function of the reciprocals of the singular values of the sampling matrix, making the statistical behavior of the smallest singular value the key factor determining reconstruction stability. As the smallest singular value collapses (approaches zero), the sampling matrix becomes increasingly ill-conditioned, leading to severe noise amplification during the inversion process and a rapid increase in the reconstruction risk, which may even cause its expectation to diverge.

In this work, we investigate two random sampling models: independent uniform sampling and jittered sampling. 
Our objective is to develop a theoretical explanation for the markedly different behaviors observed under the two sampling schemes. Specifically, we investigate why the expected reconstruction error appears to diverge for independent uniform sampling, whereas jittered sampling exhibits a finite expected error together with a sharp transition as the number of sampling points increases. We further seek to understand how these phenomena are governed by the spectral properties of the associated random Fourier sampling matrices and, ultimately, by the geometry of the sampling nodes.

\subsection{ Contribution}
We establish an exact spectral representation of the reconstruction risk for trigonometric regression, showing that the statistical performance is completely determined by the inverse singular value spectrum of the associated Fourier-Vandermonde matrix.
We identify a fundamental dichotomy between random and structured sampling. For uniformly random sampling, we prove that the expected inverse smallest singular value diverges, implying an infinite expected reconstruction risk. This demonstrates an intrinsic instability of random Fourier sampling.
For equidistant sampling, we derive the complete singular value distribution in closed form. The resulting explicit risk formula serves as the optimal benchmark for structured sampling schemes.
We develop deterministic and probabilistic perturbation theories for jittered sampling. Our estimates quantify how sampling perturbations affect the inverse singular spectrum and provide explicit finite-sample bounds on the reconstruction risk. 

We develop a spectral theory of statistical recoverability for structured trigonometric sensing systems.
By explicitly characterizing how the sampling scheme shapes the singular value spectrum of the sensing operator,
our results state that sampling scheme determines the singular value spectrum, the singular value spectrum determines statistical recoverability, and statistical recoverability, in turn, governs reconstruction risk, ill-posedness, and statistical phenomena such as spectral transitions.

\subsection{Related Work}
Random sampling of trigonometric polynomials and Fourier systems has been extensively studied in approximation theory, compressed sensing, and randomized numerical linear algebra. A major line of research concerns the stability of random Fourier sampling matrices. Seminal works by Cand\`es, Romberg, and Tao established the restricted isometry property (RIP) for randomly sampled Fourier measurements, laying the mathematical foundation of compressed sensing \cite{candes2006robust}.
Subsequently, Rauhut and coauthors developed stability estimates for random sampling of trigonometric polynomials, proving lower bounds for the smallest singular value and condition number under suitable oversampling assumptions \cite{rauhut2007random}.
Related results have also been obtained for generalized sampling and stable reconstruction from random measurements by Adcock, Hansen, and collaborators \cite{adcock2012generalized}.

From the perspective of random matrix theory, the conditioning of random sensing matrices has been investigated through the distribution of the smallest singular value and the condition number.
Fundamental results due to Edelman, Rudelson, Vershynin, Tao, Vu and Barnett provide sharp probabilistic estimates for the smallest singular value of random matrices and reveal the mechanism underlying ill-conditioning \cite{edelman1988eigenvalues,rudelson2008littlewood,tao2010random,vershynin2019high,barnett2022exponentially}.
These results offer valuable insight into the spectral behavior of random sensing operators.
Much less is known about the expectation of the reconstruction error itself.
Since the expectation is highly sensitive to rare ill-conditioned events, probabilistic stability does not necessarily imply finite expected reconstruction risk.

Another important sampling model is jittered sampling (or perturbed lattice sampling), in which sampling nodes are obtained by perturbing a uniform grid.
This model has been widely studied in nonuniform sampling theory, frame theory, and irregular Fourier analysis.
Classical results, beginning with Kadec-$1/4$ theorem, show that sufficiently small perturbations preserve the Riesz basis property of exponential systems \cite{kadets1964exact}.
These ideas were subsequently extended by Feichtinger, Gr\"ochenig, and many others to irregular sampling, frame perturbations, and stable reconstruction of bandlimited functions \cite{feichtinger1992irregular,grochenig2001foundations}.
A common conclusion of this line of research is that small perturbations preserve the stability of the sampling operator.
Our analysis reveals a fundamental dichotomy between independent uniform sampling and jittered sampling.
For independent uniform sampling, the expected inverse smallest singular value diverges, leading to an infinite expected reconstruction risk.
In contrast, the structured geometry of jittered sampling suppresses extreme clustering of sampling nodes, resulting in finite expected risk and a spectral transition that is consistent with our numerical observations.
Furthermore, several works have investigated the spectral properties of Fourier sensing matrices, particularly the behavior of their singular values, to provide theoretical explanations for the double descent phenomenon in Fourier series models \cite{farrell2011limiting,belkin2020two,xie2022overparameterization,chen2024conditioning}.

\subsection{Notation and Organization}
We introduce notation which will be used throughout the paper. 
Let $\mathbb{R}$ be the set of all real number, $\mathbb{Z}$ be the set of integers, and $\mathbb{C}$ be the set of all complex number where  $  \mathrm{i} = \sqrt{-1}$ denotes the imaginary unit.
For any given positive integer $n$, we denote $[n]=\{ 1,2,\dots,n\}$. 
Scalars, column vectors and matrices  are denoted by small letters, small bold letters  and capital bold letters, respectively, e.g., $z$, $\pmb{z} $ and $\pZ$.
The $i$-th entry of a vector $\pmb{z}$  is denoted by  $z_i$  and  the $(i, j)$-th element of a matrix $\pmb{Z}$  is denoted by $Z_{i,j}$.
 For a vector $\pmb{z}=(z_1,\ldots,z_p)^\tp \in \mbC^{p}$, let $\Vert \pmb{z}\Vert$ denote its Euclidean norm, i.e., $\Vert \pmb{z}\Vert = \sqrt{\sum_{i}\vert z_{i}\vert^2}$.
For a given matrix $\pZ \in \mbC^{n_1\times n_2}$, we use $\pZ^*$ to denote its conjugate transpose, $\pZ^\dagger$ to denote its Moore-Penrose  inverse and $\tr(\pZ)$ to denote its trace. 
Let $\Vert \pmb{Z}\Vert_F$ denote the Frobenius norm of $\pZ$, i.e., $\Vert \pmb{Z}\Vert_F = \sqrt{\sum_{ij}\vert Z_{i,j}\vert^2}$. 
If  the square matrix $\pZ$ is invertible, then we use $\pZ^{-1}$ to denote its inverse.
For two non-negative real sequences $\{a_t\}_t$ and $\{b_t\}_t$,
we write $b_t = O(a_t)$  if $b_t \leq C a_t$.
We write $f \lesssim g$ if there exists a universal constant $C$ such that $f\leq Cg$.
We write $f \asymp g$ if there exist universal constants $c,C$ such that $cg\leq f\leq Cg$.

The remainder of this paper is structured as follows. 
In Section \ref{sec:2}, we obtain an exact spectral representation of the reconstruction risk for the observation models and the associated linear regression problem. Then we present our main results on random Fourier series models for independent uniform sampling and jittered sampling. 
Section \ref{sec:5} reports numerical experiments conducted to validate the derived risk formula, thereby offering additional empirical evidence for our theoretical findings. 
Finally, Section \ref{sec:conclusions} concludes the paper.
The proofs of our main theorems are provided in Section \ref{sec:3}. 
We give the proof of all corollaries and propositions in Section \ref{sec:4}.

\section{Non-asymptotic Analysis} \label{sec:2}
Since the sampling nodes $\{t_l\}_{i=1}^m$ are random variables, the associated matrix $\pA$ 
is a random matrix. Consequently, all quantities derived from $\pA$, including its rank and singular values, are understood as random variables unless stated otherwise. In particular, $r$ denotes  the rank of the associated structured random matrix $\pA$. 
For a realization of $\pA$, we denote $\sigma_1(\pA)\geq \sigma_2(\pA)\geq \dots \geq \sigma_r(\pA)>0$ are the decreasing nonzero singular values of $\pA$.

We first derive an exact spectral representation of the reconstruction risk, reducing the statistical problem to the analysis of inverse singular values.
\begin{theorem}\label{thm1}
Assume that $\{t_l\}$ are selected from the model (a) or model (b), with $m$ and $N$ fixed. Then, we have
\begin{align}\label{points}
\mbE[\| \hx - \px \|^2] =  \frac{\epsilon^2}{m} \mbE \left[ \sum_{i=1}^r \frac{1}{\sigma_i^2(\pA)} \right] + \mbE \left[N-r \right].
\end{align}
\end{theorem}

In fact, we can obtain a deterministic and uniform lower bound for $\sum_{i=1}^{r} \frac{1}{\sigma_i^2(\pA)}$.
Since each entry of $\pA$ satisfies that $|A_{l,k}| = 1$, we know that 
\begin{align}
\| \pA\|_F^2 = \sum_{l,k} |A_{l,k}|^2 = mN.
\end{align}
Then, we have
\begin{align}
\sum_{i=1}^r \frac{1}{\sigma_i^2(\pA)} \geq \frac{r^2}{\sum_{i=1}^r\sigma_i^2(\pA)} = \frac{r^2}{\|\pA\|_F^2}= \frac{r^2}{mN}.
\end{align}

Rigorously speaking, the rank of $\pA$ satisfies $r\leq \min\{m,N\}$. 
If $t_1,\dots,t_m \in [0,1)$ are pairwise distinct, the numbers $\{e^{2\pi \mri t_l}\}_{l=1}^m$ are also pairwise distinct, and hence $\pA$ is a Vandermonde-type matrix whose maximal minors are nonvanishing, implying$r= \min \{ m, N\}$. 
Thus, when $t_1 < \dots < t_m$, 
\begin{align}
\mbE[\| \hx - \px \|^2] =
\left\{
\begin{array}{lr}
\frac{\epsilon^2}{m} \mbE \left[ \sum_{i=1}^m \frac{1}{\sigma_i^2(\pA)} \right] + N-m, &m < N,\\
\frac{\epsilon^2}{m} \mbE \left[ \sum_{i=1}^N \frac{1}{\sigma_i^2(\pA)} \right], &m \geq N.
\end{array}
\right.
\end{align}

\begin{proposition}\label{pro4}
Assume that $t_1, \dots, t_m$ are independent random variables with a distribution absolutely continuous with respect to the Lebesgue measure on $[0,1)$.
If $m \geq N$, $\pA$ has full column rank almost surely, while if $m< N$, $\pA$ has full row rank almost surely, 
and hence
\[
\sum_{i=1}^{r} \frac{1}{\sigma_i}< \infty \quad \text{almost surely}.
\]
\end{proposition}

Particularly, when $m \gtrsim N\log N$, we can get better concentration based on Matrix Bernstein Inequality \cite{tropp2012user}.
\begin{lemma}\label{Thm:Bernstein} 
Consider a finite sequence $\{\pX_k\}$ of independent, random, self-adjoint matrices with dimension $d$. Assume that
$\mbE \pX_k = 0$ and $\lambda_{\max}(\pX_k) \leq R$ almost surely.
Compute the norm of the total variance,
\[
\sigma^2:= \left\|\sum_k\mbE \pX^2_k\right\|.
\]
Then, the following chain of inequalities holds for all $t \geq 0$:
\[
\mathbb{P}\left[\lambda_{\max}\left(\sum_k \pX_k\right)\geq t\right] \leq d \cdot \exp\left(\frac{-t^2/2}{\sigma^2+Rt/3}\right).
\]
\end{lemma}

\begin{proposition} \label{pro3}
If $m \gtrsim N\log N$, then there exist constants $c, C >0$ that the following holds
\[
\mathbb{P}\left[ \|\pA^*\pA-\pI\|\leq C\sqrt{\frac{N}{m}} \right] \geq 1-2e^{-cN}.
\]
Moreover, we have
\[
0<1- C\sqrt{\frac{N}{m}} \leq \sigma_N^2(\pA) \leq \sigma_1^2(\pA) \leq 1+C\sqrt{\frac{N}{m}}
\]
with probability at least $1- 2e^{-cN}$.
\end{proposition}

Based on Proposition \ref{pro3}, we can get $r=N$ when $m \gtrsim N\log N$, and 
\begin{align}
 \frac{N}{1+C\sqrt{\frac{N}{m}}} \leq \frac{N}{\sigma_1^2(\pA)} \leq \sum_{i=1}^N \frac{1}{\sigma_i^2(\pA)} \leq \frac{N}{\sigma_N^2(\pA)} \leq \frac{N}{1-C\sqrt{\frac{N}{m}}}
\end{align}
with probability at least $1- 2e^{-cN}$.

The following proposition shows that all singular values of $\pA$ are bounded away from both zero and infinity when the sampling nodes do not contain large separations. We first give the definition of torus distance \cite{candes2014towards}:
\begin{align}
\dist_{T}(t_1,t_2) = \min_{n\in \mathbb{Z}}|t_1-t_2+n|, \quad \text{for } t_1, t_2 \in [0,1). 
\end{align}

\begin{proposition}[Theorem 5 in \cite{grochenig1993discrete}]\label{pro1}
Suppose that $0 \leq t_1 < \dots < t_m <1$ and define the maximal separation $\Delta_{\max}$ as \footnote{Here, we set $t_{m+1} = t_1$.}
\begin{align}
\Delta_{\max}: = \max_{i\in [m]} \dist_T(t_{i+1},t_i) < d,
\end{align}
$\sigma_{1}(\pA) \leq 1+\frac{\Delta_{\max}}{d}<2$ and $\sigma_{r}(\pA) \geq 1-\frac{\Delta_{\max}}{d}>0$.
\end{proposition}

Proposition \ref{pro1} shows that the maximal separation condition provides explicit upper and lower bounds for the singular values of the matrix $\pA$. However, this analysis relies on the assumption that the sampling nodes $t_1,\dots,t_m$ are sufficiently well distributed on $[0,1)$, namely that the maximal separation satisfies $\Delta_{\max}<d$.

For randomly distributed sampling nodes, clustering effects may occur with positive probability, meaning that some nodes $t_l$ can become arbitrarily close to each other. In such situations, the assumptions of Proposition \ref{pro1} are no longer sufficient to characterize the conditioning of $\pA$. More importantly, the behavior of the smallest singular value is governed by the minimal separation. Hence, upper bounds for $\sigma_{r}(\pA)$ become more relevant for understanding the possible ill-conditioning of the matrix. The following proposition illustrates that near-collision implies small singular value.
\begin{proposition}
\label{pro:sigmin}
For pairwise distinct nodes $\{t_l\}_{l=1}^{m} \subset [0,1)^m$, we define the minimal separation $\Delta_{\min}$ as
\begin{align}
\Delta_{\min}: = \min_{i\ne j}  \dist_T(t_{i},t_j).
\end{align}
There exists a constant $C>0$ such that
\begin{align}
\sigma_{r}(\pA) \leq CN^{3/2}\Delta_{\min}.
\end{align}
Especially, for any nodes $\{t_l\}_{l=1}^{m} \subset [0,1)^m$, the minimal separation $\Delta_{\min}$ may be zero.
There exists a constant $C>0$ such that
\begin{align}
\sigma_{r}(\pA) \leq CN^{3/2}\Delta^+_{\min},
\end{align}
where $\Delta^+_{\min}$ denotes the smallest nonzero separation with 
\begin{align}
\Delta^+_{\min}: = \min_{i\ne j, t_i\ne t_j}  \dist_T(t_{i},t_j).
\end{align}
\end{proposition}

The dependence of the smallest singular value on the minimum node separation has been extensively studied in the context of super-resolution and the conditioning of Vandermonde matrices; see, e.g., Moitra \cite{moitra2015super} and Kunis and Nagel \cite{kunis2020smallest}. Our contribution is to connect this spectral instability with statistical recoverability. Specifically, the above estimate shows that geometric degeneracy of the sampling nodes directly translates into instability of linear reconstruction through the inverse singular value spectrum.

\subsection{Sample model (a)}
In this subsection,  we consider the case where $t_1,\dots , t_m$ are i.i.d. random variables, each of which is uniformly distributed over $[0, 1)$. 
We prove that uniformly random sampling exhibits an intrinsic instability: the expectation of the inverse smallest singular value diverges, implying an infinite expected reconstruction risk. It reveals a sharp distinction between random and structured sampling.

We need to give the estimation of $\mbE \left[ \sum_{i=1}^r \frac{1}{\sigma_i^2(\pA)} \right]$. The key is to analyze whether $\mbE \frac{1}{\sigma_r^2(\pA)}$ converges.
The fact that random variables are finite almost everywhere does not mean that their expected value is finite. For example, we set $X(\omega) = \frac{1}{\omega}$, $\omega \sim \text{Unif}(0,1)$, then $X < \infty$ a.s., but $\mbE X = +\infty$.
The key observation is that almost sure finiteness does not imply integrability; heavy tails may arise from events where sampling locations are nearly colliding.
Based on Proposition \ref{pro:sigmin}, we can get 

\begin{theorem}\label{Thm:Divergence}
The $t_1,\dots , t_m$ are i.i.d. random variables, each of which is uniformly distributed over $[0, 1)$.
Then we have
\[
\mathbb{E}\left[\frac{1}{\sigma_{r}^2(\pA)}\right] =+\infty.
\]
Furthermore, it means that $\mbE[\| \hx - \px \|^2] $ diverges.
\end{theorem}
This fundamental obstruction shows that random sampling alone cannot guarantee stable reconstruction in expectation, and it motivates the need for either deterministic node perturbations.

\subsection{Sample model (b)}
In this subsection, we consider that the sampling nodes are obtained as random perturbations of a uniform grid:
\[
t_l=\frac{l-1}{m}+\delta_l,\qquad l=1,\dots,m,
\]
where $\delta_1,\dots,\delta_m$ are i.i.d. random variables uniformly distributed on $(0,T)$.
Independent uniform sampling allows arbitrarily large gaps and clusters with positive probability, whereas jittered sampling preserves the global regularity of an underlying uniform grid while introducing only local perturbations.

We first give the Cumulative Distribution Function (CDF) under this case for different $T$:
\begin{proposition}\label{prop3_8}
For fixed points $t_1 = \delta_1$ and $t_2 = \frac{1}{m} + \delta_2$ with $\delta_1, \delta_2 \overset{i.i.d.}{\sim} \Unif(0,T)$ and $T\leq \frac{1}{m}$, we set $D = t_2-t_1$ and give the CDF of $D$:
 \[
\mathbb{P}[ D \leq \varepsilon] = \left\{
\begin{array}{lr}
0, \quad &\varepsilon < \frac{1}{m}-T,\\
\frac{(\varepsilon -\frac{1}{m}+T)^2}{2 T^2}, \quad&\frac{1}{m}-T \leq  \varepsilon \leq \frac{1}{m},\\
1 - \frac{(\varepsilon -\frac{1}{m} -T)^2}{2 T^2}, \quad&\frac{1}{m} < \varepsilon \leq \frac{1}{m}+T,\\
1, \quad&\frac{1}{m}+T < \varepsilon.
\end{array}
\right.
\]
\end{proposition}

\subsubsection{Sample model (b) with $\delta_l \in (0,T)$ with $T=\frac{1}{m}$}
Based on Proposition \ref{pro:sigmin} and Proposition \ref{prop3_8}, we can get 
\begin{theorem}\label{Thm2}
$t_l =  \frac{l-1}{m} + \delta_l, l =1,\dots,m$, where $\delta_l\overset{i.i.d.}{\sim} \Unif(0,\frac{1}{m})$.
Then we have
\[
\mathbb{E}\left[\frac{1}{\sigma_{r}^2(\pA)}\right] =+\infty.
\]
Furthermore, it means that $\mbE[\| \hx - \px \|^2] $ diverges.
\end{theorem}

\subsubsection{Sample model (b) with $\delta_l \in (0,T)$ with $T<\frac{1}{m}$}
\paragraph{\bf{Equidistant sampling nodes}}
Before we analyze random sampling model (b) with $T<\frac1m$, we first present a deterministic result for the discrete Fourier matrix corresponding to equidistant sampling nodes.
Since $t_l = \frac{l-1}{m}, l=1,\dots,m$ are pairwise distinct, the numbers $\{e^{2\pi \mri t_l}\}_{l=1}^m$ are also pairwise distinct. Hence $\pA$ is a Vandermonde-type matrix whose maximal minors are nonvanishing, implying $\rank(\pA)=\min\{m,N\}$.

The equidistant Fourier sampling case admits an explicit singular value characterization due to the underlying discrete Fourier structure \cite{christensen2003introduction,potts2003fast}.
The following theorem gives the resulting explicit singular value characterization for equidistant Fourier--Vandermonde matrices.
It serves as the baseline for the subsequent analyses involving perturbations. 
While this spectral property is classical, its implication for statistical reconstruction risk has not been explicitly characterized. By combining this exact spectrum with the pseudoinverse estimator, we obtain a closed-form risk formula and use it as a benchmark for perturbed sampling schemes.

\begin{theorem} \label{Equidistant}
Set $\Gamma = \{-q,\dots,0, \dots,q\}$ and  $N = 2q+1$. 
Let $\pA \in \mathbb{C}^{m \times N}$ with entries
\[
A_{l,k} = e^{2\pi \mri k t_l}, \quad t_l = \frac{l-1}{m},\; l=1,\dots,m,
\]
where $k \in \Gamma$.
Let $r = \operatorname{rank}(\pA) = \min\{m,N\}$.
Then the largest and smallest nonzero singular values of $\pA\pA^*$ are given by
\[
\sigma_1(\pA\pA^*) = 
\begin{cases}
m, & m \ge N,\\[6pt]
\displaystyle m\Bigl\lceil \frac{N}{m} \Bigr\rceil, & m < N,
\end{cases}
\qquad
\sigma_r(\pA\pA^*) = 
\begin{cases}
m, & m \ge N,\\[6pt]
\displaystyle m\Bigl\lfloor \frac{N}{m} \Bigr\rfloor, & m < N.
\end{cases}
\]
Moreover, each singular value equals $m$ when $m\geq N$; each singular value can take only two values:
\[
\begin{cases}
m(a+1), & \text{for } b \text{ residue classes}, \\[4pt]
ma, & \text{for } m-b \text{ residue classes},
\end{cases}
\]
when $m < N$ and $N = am +b$ with $a = \lfloor N/m \rfloor$, $0\leq b<m$.
\end{theorem}
Then, we can get the following two corollaries about $\sum_{i=1}^{r}\frac1{\sigma_i(\pA\pA^*)}$.
\begin{corollary}\label{cor3_11}
Let $N=am+b$ with $a=\Bigl\lfloor N/m \Bigr\rfloor$ and $ 0\le b<m$.
Then
\[
\sum_{i=1}^{r}\frac1{\sigma_i(\pA\pA^*)}
=\begin{cases}
\displaystyle \frac{N}{m}, & m\ge N, \\[10pt]
\displaystyle \frac{m+N-2b}{ma(a+1)}, & m<N.
\end{cases}
\]
\end{corollary}

\begin{corollary}\label{cor3_12}
Assume $m<N$. Then
\[
\frac{1}{\Bigl\lfloor\frac{N}{m}\Bigr\rfloor+1} \le \sum_{i=1}^{r}\frac{1}{\sigma_i(\pA\pA^*)} \le \frac{1}{\Bigl\lfloor\frac{N}{m}\Bigr\rfloor}.
\]
\end{corollary}
Based on Corollary \ref{cor3_11} and $r = \min\{m,N\}$, we derive a closed-form expression for the reconstruction risk.
\begin{theorem} \label{Equidistant2}
At the same setting of Theorem \ref{Equidistant},
we have
\begin{align}\label{points2}
\mbE[\| \hx - \px \|^2]=\begin{cases}
\displaystyle  \frac{N}{m^2}\epsilon^2, & m\ge N; \\
\displaystyle  \frac{m+N-2b}{a(a+1)m^2}\epsilon^2 + N-m, & m<N.
\end{cases}
\end{align}
\end{theorem}

\paragraph{\bf{Non-equidistant sampling nodes}}
Now, we provide a necessary discussion on the case of non-equidistant sampling nodes:
\[
t_l=\frac{l-1}{m}+\delta_l, \qquad l=1,\dots,m.
\]

\paragraph{\bf{Sample model (b) with small $T$ (deterministic perturbations)}}
We consider deterministic perturbations of non-equidistant sampling nodes, and derive a deterministic upper bound for 
$1/\sigma_r(\pA\pA^*)$, which is obtained based on Weyl's inequality.
Here, we consider 
\[
\delta_l \in (0, T)
\] 
with $0<T<1/m$. Since $\frac{l-1}{m} \le t_l < \frac{l}{m}$,
the nodes satisfy $t_1<t_2<\cdots<t_m$,
which means that the nodes are pairwise distinct and therefore
\[
r = \rank(\pA)=\min\{m,N\}.
\]
Let $\tilde{\pA}$ denote the equidistant matrix with entries $\tilde{A}_{l,k} = e^{2\pi \mri k \frac{l-1}{m}}$, then 
\[
\pA_{l,k} = e^{2\pi \mri k \delta_l} \cdot \tilde{A}_{l,k}.
\]
For convenience, we set $\theta := \sigma_{r}(\tilde{\pA}\tilde{\pA}^*)$ and $\alpha := \|\tilde{\pA}\| = \sqrt{\sigma_{1}(\tilde{\pA}\tilde{\pA}^*)}$. 
Define $\pE = \pA - \tilde{\pA}$, i.e. $E_{l,k} = \tilde{A}_{l,k}\bigl(e^{2\pi \mri k\delta_l} - 1\bigr)$. 
Since $|\delta_l| \le T$ and $|k| \le q \le N/2$, we have pointwise
\[
|e^{2\pi \mri k\delta_l} - 1| = 2|\sin(\pi k\delta_l)| \le 2\pi |k| T \le \pi N T.
\]
Hence, we can get
\begin{align}\label{eq:upperound1_general}
\|\pE\| \le \|\pE\|_F \le \sqrt{mN}\,(\pi N T) = \pi\sqrt{m}\,N^{3/2}T =: \beta_{\max}.
\end{align}
Let $\pP = \tilde{\pA} \pE^* + \pE\tilde{\pA}^* + \pE\pE^*$. Then we have 
\begin{align}\label{eq:upperound2_general}
\|\pP\| \le 2\|\tilde{\pA}\| \|\pE\| + \|\pE\|_2^2 \le 2\alpha\beta_{\max} + \beta_{\max}^2 =: \eta.
\end{align}
By Weyl's inequality, we have $|\sigma_i(\pA\pA^*) - \sigma_i(\tilde{\pA}\tilde{\pA}^*)| \le \|\pP\|$,
which gives
\[
\sigma_r(\pA\pA^*) \ge \sigma_r(\tilde{\pA}\tilde{\pA}^*) - \|\pP\| \ge \theta - \eta,
\]
\[
\sigma_1(\pA\pA^*) \le \sigma_1(\tilde{\pA}\tilde{\pA}^*) + \|\pP\| \le \alpha^2 + \eta.
\]

If $m< N$, we know that $\alpha =\sqrt{ m\Bigl\lceil N/m \Bigr\rceil }$ and $\theta = m\Bigl\lfloor N/m \Bigr\rfloor$ based on Theorem \ref{Equidistant}. 
We set $T < (\sqrt{2N} - \sqrt{N+m}) / \pi\sqrt{m}\,N^{3/2}$ 
to  guarantee $\theta \ge m > \eta$. Under this condition, $\sigma_r(\pA\pA^*) > 0$ almost surely and the following bounds hold pathwise.

From the above inequalities, for every realization, we have
\[
\frac{1}{\sigma_r(\pA\pA^*)} \le \frac{1}{\theta - \eta},\qquad
\frac{1}{\sigma_1(\pA\pA^*)} \ge \frac{1}{\alpha^2 + \eta}.
\]
Since the right‑hand sides are deterministic constants, taking expectations yields:

\begin{proposition}\label{pro7}
Under the assumptions $m < N$ and
\[
T < \min\left\{\frac{1}{m},\ \frac{\sqrt{2N} - \sqrt{N+m}}{\pi\sqrt{m} N^{3/2}}\right\},
\]
the following bounds hold:
\[
\frac{1}{\sigma_r(\pA\pA^*)} \le
\frac{1}{m\left\lfloor\frac{N}{m}\right\rfloor - 2\pi m \sqrt{\left\lceil\frac{N}{m}\right\rceil} N^{3/2}T - \pi^2 m N^3 T^2}
\]
and
\[
\frac{1}{\sigma_1(\pA\pA^*)} \ge
\frac{1}{m\left\lceil\frac{N}{m}\right\rceil +2\pi m \sqrt{\left\lceil\frac{N}{m}\right\rceil} N^{3/2}T + \pi^2 m N^3 T^2}.
\]
\end{proposition}

Similarly, we know that $\alpha =\sqrt{ m}$ and $\theta = m$ based on Theorem \ref{Equidistant} if $m \geq N$.
We set $T < 1 / (\sqrt{2}+1)\pi N^{3/2}$  to  guarantee $\theta = m  > \eta$. Then we have the following results.
\begin{proposition}\label{pro8}
Under the assumptions $m \geq N$ and
\[
T < \min\left\{\frac{1}{m},\ \frac{\sqrt{2}-1}{\pi N^{3/2}}\right\},
\]
the following bounds hold:
\[
\frac{1}{\sigma_r(\pA\pA^*)}\le
\frac{1}{m - 2\pi m N^{3/2}T - \pi^2 m N^3 T^2}
\]
and
\[
\frac{1}{\sigma_1(\pA\pA^*)} \ge
\frac{1}{m + 2\pi m N^{3/2}T + \pi^2 m N^3 T^2}.
\]
\end{proposition}
Using the elementary inequality
\[
\frac{r}{\sigma_1(\pA\pA^*)} \le \sum_{i=1}^{r}\frac{1}{\sigma_i(\pA\pA^*)} \le \frac{r}{\sigma_r(\pA\pA^*)},
\]
together with Propositions \ref{pro7} and \ref{pro8}, we obtain the following result.

\begin{theorem}\label{thm:sum_inverse_singular}
Set $\Gamma = \{-q,\dots,0, \dots,q\}$ and  $N = 2q+1$.
Let $\pA\in\mathbb C^{m\times N}$ be the Fourier--Vandermonde matrix with entries
\[
A_{l,k} = e^{2\pi \mri k t_l},
\qquad
t_l=\frac{l-1}{m}+\delta_l,
\qquad
l=1,\dots,m,
\]
where $k \in \Gamma$ and $0<\delta_l<T$. 
Since $T<1/m$, the nodes satisfy $t_1<t_2<\cdots<t_m$, and therefore $r=\rank(\pA)=\min\{m,N\}$.

\noindent
{\rm(i) Case $m\ge N$.}
Assume that $T< \min\left\{ \frac{1}{m}, \frac{1}{(\sqrt2+1)\pi N^{3/2}} \right\}$. 
Then
\[
\frac{N}{m\left( 1 +2\pi N^{3/2}T +\pi^2N^3T^2 \right)}
\le
\sum_{i=1}^{r} \frac{1}{\sigma_i(\pA\pA^*)}
\]
and 
\[
\sum_{i=1}^{r} \frac{1}{\sigma_i(\pA\pA^*)}
\le
\frac{N}{m\left( 1 -2\pi N^{3/2}T -\pi^2N^3T^2\right)}.
\]

\noindent
{\rm(ii) Case $m<N$.} Assume that $T< \min\left\{ \frac{1}{m}, \frac{\sqrt{2N}-\sqrt{N+m}}{\pi\sqrt m\,N^{3/2}}\right\}$. 
Then
\[
\frac{1}{ \left\lceil\frac Nm\right\rceil + 2\pi \sqrt{\left\lceil\frac Nm\right\rceil} N^{3/2}T + \pi^2N^3T^2}
\le \sum_{i=1}^{r} \frac{1}{\sigma_i(\pA\pA^*)}
\]
and
\[
\sum_{i=1}^{r} \frac{1}{\sigma_i(\pA\pA^*)}
\le \frac{1}{\left\lfloor\frac Nm\right\rfloor - 2\pi \sqrt{\left\lceil\frac Nm\right\rceil} N^{3/2}T - \pi^2N^3T^2}.
\]
\end{theorem}

Particularly, we know that the nonzero singular values of $\tilde{\pA}\tilde{\pA}^*$ are given by
\[
\sigma_i(\tilde{\pA}\tilde{\pA}^*) =
\begin{cases}
m(a+1), & \text{for } b \text{ indices}, \\[4pt]
ma, & \text{for } m-b \text{ indices},
\end{cases}
\]
from Theorem~\ref{Equidistant}.  Propositions~\ref{pro7} and \ref{pro8} imply that
\[
\bigl| \sigma_i(\pA\pA^*) - \sigma_i(\tilde{\pA}\tilde{\pA}^*) \bigr| \le \eta, \qquad i=1,\dots,m.
\]

If $m<N$, by Theorem~\ref{Equidistant},
the nonzero singular values of $\tilde{\pA}\tilde{\pA}^*$ consist of $ma$ with multiplicity $m-b$, and $m(a+1)$ with multiplicity $b$.
Hence
\begin{align*}
\sum_{i=1}^{m} \frac{1}{\sigma_i(\pA\pA^*)} \ge \sum_{i=1}^{m} \frac{1}{\sigma_i(\tilde{\pA}\tilde{\pA}^*)+\eta} &=
\frac{m-b}{ma+\eta} + \frac{b}{m(a+1)+\eta} \\
&= \frac{m(N+m+\eta-2b)}{(ma+\eta)[m(a+1)+\eta]} 
\end{align*}
and
\begin{align*}
\sum_{i=1}^{m} \frac{1}{\sigma_i(\pA\pA^*)} \le \sum_{i=1}^{m} \frac{1}{\sigma_i(\tilde{\pA}\tilde{\pA}^*)-\eta} &=
\frac{m-b}{ma-\eta} + \frac{b}{m(a+1)-\eta} \\
&= \frac{m(N+m-\eta-2b)}{(ma-\eta)[m(a+1)-\eta]} 
\end{align*}
for $T< \min\left\{ \frac{1}{m}, \frac{\sqrt{2N}-\sqrt{N+m}}{\pi\sqrt m\,N^{3/2}}\right\}$.
Here, $\eta = 2\alpha\beta_{\max} + \beta_{\max}^2$, $\alpha = \sqrt{ m\Bigl\lceil \frac{N}{m} \Bigr\rceil }$, $\beta_{\max} = \pi\sqrt{m} N^{3/2}T$, $a = \left\lfloor\frac Nm\right\rfloor$ and $b = N-am$.

If $m\ge N$, then by Theorem~\ref{Equidistant}, all nonzero singular values of $\tilde{\pA}\tilde{\pA}^*$ are equal to $m$,
and there are exactly $N$ such singular values.
Therefore
\[
\frac{N}{m + \eta} \le \sum_{i=1}^{N} \frac{1}{\sigma_i(\pA\pA^*)} \le \frac{N}{m - \eta},
\]
for $T< \min\left\{ \frac{1}{m}, \frac{1}{(\sqrt2+1)\pi N^{3/2}} \right\}$.
Here, $\eta = 2\alpha\beta_{\max} + \beta_{\max}^2$, $\alpha = \sqrt{ m}$ and $\beta_{\max} = \pi\sqrt{m} N^{3/2}T$.

Therefore, we obtain the following theorem:

\begin{theorem}
\label{thm:improved_sum_inverse} Set $\Gamma = \{-q,\dots,0, \dots,q\}$ and  $N = 2q+1$.
Let $\pA \in \mathbb{C}^{m \times N}$ with entries
\[
A_{l,k} = e^{2\pi \mri k t_l}, \quad t_l = \frac{l-1}{m}+\delta_l,\ l=1,\dots,m,
\]
where $k \in \Gamma$ and $\delta_l \in  (0,T)$, $ 0<T<\frac{1}{m}$. 
Since $T<1/m$, the nodes satisfy $t_1<t_2<\cdots<t_m$, and therefore $r=\rank(\pA)=\min\{m,N\}$.
Write $N=am+b, a=\left\lfloor N/m\right\rfloor, 0\le b<m$. 
Then the following holds.

\noindent
{\rm(i) Case $m\ge N$.}
Assume that $\eta = 2\alpha\beta_{\max} + \beta_{\max}^2$, $\alpha = \sqrt{ m}$, $\beta_{\max} = \pi\sqrt{m} N^{3/2}T$ and 
\[
T< \min\left\{ \frac{1}{m}, \frac{1}{(\sqrt2+1)\pi N^{3/2}} \right\}.
\]
Then we have
\[
 \frac{N}{m(m + \eta)} \epsilon^2 \leq \mbE[\| \hx - \px \|^2] \leq \frac{N}{m(m - \eta)} \epsilon^2 .
\]

\noindent
{\rm(ii) Case $m<N$.} Assume that $\eta = 2\alpha\beta_{\max} + \beta_{\max}^2$, $\alpha = \sqrt{ m\Bigl\lceil N/m \Bigr\rceil }$, $\beta_{\max} = \pi\sqrt{m} N^{3/2}T$ and
\[
T< \min\left\{ \frac{1}{m}, \frac{\sqrt{2N}-\sqrt{N+m}}{\pi\sqrt m\,N^{3/2}}\right\}.
\]
Then we have
\begin{align*}
 \frac{N+m+\eta-2b}{(ma+\eta)[m(a+1)+\eta]} \epsilon^2 + N-m \leq \mbE[\| \hx - \px \|^2] \quad \quad\\
 \leq  \frac{N+m-\eta-2b}{(ma-\eta)[m(a+1)-\eta]} \epsilon^2+ N-m.
\end{align*}
\end{theorem}

\paragraph{\bf Sample model (b) with small $T$ (random perturbations)}
In contrast to the deterministic analysis above, the randomness allows us to exploit concentration inequalities and obtain substantially sharper estimates for $\mathbb E\!\left[ \sum_{i=1}^{r} \frac{1}{\sigma_i(\pA\pA^*)} \right]$. 
We now consider random perturbations of the equidistant sampling nodes $t_l=\frac{l-1}{m}+\delta_l,  l=1,\ldots,m$ with
$\delta_l \overset{\mathrm{i.i.d.}}{\sim} \mathrm{Unif}(0,T), 0<T<1/m$.

We also set $A_{l,k} = e^{2\pi \mri k t_l}$ with $k\in\Gamma=\{-q,\ldots,q\}$ and $N=2q+1$.
Denote by $\tilde{\pA}$ the equidistant Fourier matrix with entries $\tilde{A}_{l,k} = e^{2\pi \mri k(l-1)/m}$, 
and define $\pE = \pA-\tilde{\pA}$ and $\pP = \pA\pA^* - \tilde{\pA}\tilde{\pA}^* = \tilde{\pA}\pE^* + \pE\tilde{\pA}^* + \pE\pE^*$.
Since
\[
|e^{2\pi \mri k\delta_l}-1| = 2|\sin(\pi k\delta_l)| \le 2\pi |k|T \le \pi NT,
\]
we obtain $\|\pE\| \le \|\pE\|_F \le \pi \sqrt{m} N^{3/2}T$. 
Define $\beta_{\max} = \pi \sqrt{m} N^{3/2}T$. 
Furthermore, we have
\[
\|\pP\| \le 2\|\tilde{\pA}\| \|\pE\| + \|\pE\|^2.
\]
Let $\alpha = \|\tilde{\pA}\|$ and $\eta = 2\alpha\beta_{\max} + \beta_{\max}^2$. We then have $\|\pP\| \le \eta$. 
Define
\[
\mu_k = \mathbb{E}\Big[ |e^{2\pi \mri k\delta}-1|^2 \Big] =
\begin{cases}
0, &k=0, \\[4pt]
2- \frac{\sin(2\pi kT)}{\pi kT}, &k\neq 0,
\end{cases}
\]
and $S_1 = \sum_{k=-q}^{q}\mu_k$. 
Since $
\|\pE\|_F^2= \sum_{l=1}^{m} \sum_{k=-q}^{q} |e^{2\pi \mri k\delta_l}-1|^2$, 
we have $\mathbb{E}\|\pE\|_F^2 = mS_1$. 
Consequently,
\[
\mathbb{E}\|\pP\| \le 2\alpha\sqrt{mS_1} + mS_1.
\]
For any $t_0\in(0,\eta)$, 
define $y_{t_0} = -\alpha+\sqrt{\alpha^2+t_0}$. 
Since $\|\pP\| \le 2\alpha\|\pE\|_F+\|\pE\|_F^2$, 
the event $\{\|\pP\|>t_0\}$ implies $\|\pE\|_F>y_{t_0}$, or equivalently $\|\pE\|_F^2>y_{t_0}^2$. 
Let $ Z_l=\sum_{k=-q}^{q} |e^{2\pi \mri k\delta_l}-1|^2$. 
Then, we have
\[
\|\pE\|_F^2 = \sum_{l=1}^{m}Z_l .
\]
The variables $Z_l$ are independent and satisfy $0\le Z_l\le 4N$. 
Let $v = m \mathrm{Var}(Z_1)$. 
Assume $y_{t_0}^2>mS_1$, 
we have
\[
\mathbb{P}[\|\pP\|>t_0] \le
\mathbb{P}\left[\|\pE\|_F^2>y_{t_0}^2 \right] = \mathbb{P}\left[ \sum_{l=1}^{m} Z_l > y_{t_0}^2 \right].
\]
Since $\mathbb{E} Z_l=S_1$, we obtain
\[
\mathbb{P}[\|\pP\|>t_0] \le
\mathbb{P} \left[ \sum_{l=1}^{m}(Z_l-S_1) > y_{t_0}^2-mS_1 \right].
\]
Then we can use  Bernstein's inequality to get $\mathbb{P}[\|\pP\|>t_0] \le p_{t_0}$, 
where
\[
p_{t_0} = \exp\left(-\frac{(y_{t_0}^2-mS_1)^2}{2v+\frac{8}{3}N(y_{t_0}^2-mS_1)}\right).
\]
On the other hand,
\[
y_{\eta} = -\alpha + \sqrt{\alpha^2 + 2\alpha\beta_{\max} + \beta_{\max}^2} = \beta_{\max} = \pi \sqrt{m} N^{3/2}T.
\]
Using $|e^{\mri x} -1|\leq |x|$, we can get $\mu_k = \mathbb{E}\Big[ |e^{2\pi \mri k\delta}-1|^2 \Big] \leq 4\pi^2k^2T^2$. Then
\[
mS_1 = m\sum_{k=-q}^{q}\mu_k \le 4\pi^2 m T^2 \sum_{k=-q}^q k^2 = \frac{\pi^2 m T^2 N(N^2-1)}{3}.
\]
Consequently, we can get
\[
\frac{mS_1}{y_\eta^2} \le \frac{N^2-1}{3N^2} < 1,
\]
which implies $mS_1<y_\eta^2$. 
Since $y_t$ is strictly increasing in $t$, $y_t^2 > mS_1$ is equivalent to $t>mS_1+2\alpha\sqrt{mS_1}$. 
Therefore, 
\[
y_{t_0}^2>mS_1
\]
holds whenever $t_0 \in (mS_1+2\alpha\sqrt{mS_1}, \eta)$. 

We can now derive the explicit bounds for $\mbE[\| \hx - \px \|^2]$.

\begin{theorem}
\label{thm:random_smallT_explicit} Set $\Gamma = \{-q,\dots,0, \dots,q\}$ and  $N = 2q+1$.
Let $A \in \mathbb{C}^{m \times N}$ with entries
\[
\pA_{l,k} = e^{2\pi \mri k t_l}, \quad t_l = \frac{l-1}{m}+\delta_l,\ l=1,\dots,m,
\]
where $k \in \Gamma$ and $\delta_l \in  (0,T)$, $0<T<\frac{1}{m}$. 
Since $T<1/m$, the nodes satisfy $t_1<t_2<\cdots<t_m$, and therefore $r=\rank(\pA)=\min\{m,N\}$.
Write $N=am+b, a=\left\lfloor\frac Nm\right\rfloor, 0\le b<m$. 
Let $\beta_{\max}= \pi\sqrt{m} N^{3/2}T$. Define
\[
\mu_k =
\begin{cases}
0,
&
k=0,
\\[4pt]
2-
\dfrac{\sin(2\pi kT)}
{\pi kT},
&
k\neq0,
\end{cases}
\]
and $S_1=\sum_{k=-q}^{q}\mu_k$. 
Then the following  holds.

\noindent
{\rm(i) Case \(m\ge N\).}
In this case, we set $\eta = 2\sqrt{m}\beta_{\max} + \beta_{\max}^2$. 
For any $t_0 \in (mS_1+2m\sqrt{S_1}, \eta)$, we define $y_{t_0} = -\sqrt{m}+\sqrt{m+t_0}$ 
and 
\[
p_{t_0} = \exp \left( -\frac{\bigl(y_{t_0}^2-mS_1\bigr)^2}{N \left(8mS_1+ \frac{8}{3} \bigl(y_{t_0}^2-mS_1\bigr) \right)}\right).
\]
If  $T < \frac{\sqrt{2}-1}{\pi N^{3/2}}$, 
then we have
\[
\frac{N}{m^2(1+\sqrt{S_1})^2} \epsilon^2 \le \mbE[\| \hx - \px \|^2] \le \frac{N}{m} \left(\frac{1-p_{t_0}}{m-t_0}+\frac{p_{t_0}}{m-\eta}\right)\epsilon^2.
\]

\noindent
{\rm(ii) Case \(m<N\).} In this case, we set $\eta = 2\sqrt{m(a+1)}\,\beta_{\max} + \beta_{\max}^2$. 
For any $t_0 \in \Bigl(mS_1+2m\sqrt{(a+1)S_1}, \eta \Bigr)$, 
define $y_{t_0} = -\sqrt{m(a+1)} + \sqrt{m(a+1)+t_0}$, 
and
\[
p_{t_0} =\exp \left( -\frac{\bigl(y_{t_0}^2-mS_1\bigr)^2}{N\left(8mS_1+ \frac{8}{3}\bigl(y_{t_0}^2-mS_1\bigr)\right)}\right).
\]
If $T < \frac{\sqrt{2a+1} - \sqrt{a+1}}{\pi N^{3/2}}$, 
then we have
\begin{align*}
\frac{m\Bigl(a+1+2\sqrt{(a+1)S_1}+S_1\Bigr)-b}{m^2\Bigl(a+2\sqrt{(a+1)S_1}+S_1\Bigr)\Bigl(a+1+2\sqrt{(a+1)S_1}+S_1\Bigr)}\epsilon^2 + N-m\\
\le
\mbE[\| \hx - \px \|^2]
\end{align*}
and
\begin{align*}
\mbE[\| \hx - \px \|^2]
\le&
\frac{\epsilon^2}{m}(m-b)\left( \frac{1-p_{t_0}}{ma-t_0}+ \frac{p_{t_0}}{ma-\eta} \right)\\
&+
\frac{b\epsilon^2}{m}\left(\frac{1-p_{t_0}}{m(a+1)-t_0}+ \frac{p_{t_0}}{m(a+1)-\eta}\right) + N-m.
\end{align*}
\end{theorem}

\section{Numerical Experiments} \label{sec:5}
In this section, we construct numerical experiments based on the  two sampling models to demonstrate our conclusions.
\subsection{Sample (a)} \label{sec:num1}
To investigate the empirical behavior of $1/\sigma^2_r(\pA)$, we perform Monte Carlo simulations for a random Fourier matrix $\pA\in\mathbb{C}^{m\times N}$ with $N=41$, $m=33$ ($m/N\approx 0.8$), and increasing numbers of independent trials. Table~\ref{tab:mean_median} summarizes the empirical mean, median, and maximum of $1/\sigma_r(A)^2$ for different sample sizes.

\begin{table}[!h]
\vspace{-0.5cm}  
\setlength{\abovecaptionskip}{0.1cm}   
\centering
\caption{Empirical mean, median, and maximum of $1/\sigma_r^2(\pA)$ for different numbers of trials.}
\label{tab:mean_median}
\begin{tabular}{|c|c|c|c|}
\hline
Number of trials & Mean & Median & Maximum \\ \hline
$10$ & $2.51577\times 10^8$ & $5.4750\times 10^{5}$ & $2.0938\times 10^{9}$ \\ \hline
$10^2$ & $1.4267\times 10^{17}$ & $1.1782\times 10^{6}$ & $1.4266\times 10^{19}$ \\ \hline
$10^3$ & $1.3380\times 10^{21}$ & $5.6998\times 10^{5}$ & $1.2980\times 10^{24}$ \\ \hline
$10^4$ & $9.1525\times 10^{22}$ & $7.8869\times 10^{5}$ & $7.3821\times 10^{26}$ \\ \hline
$10^5$ & $1.9479\times 10^{24}$ & $7.8721\times 10^{5}$ & $2.2461\times 10^{29}$ \\ \hline
\end{tabular}
\vspace{-0.3cm}
\end{table}
As the number of trials increases from $10$ to $10^5$, the empirical mean of $1/\sigma_r^2(\pA)$ grows dramatically, from approximately $\sim10^8$ to $\sim10^{24}$, whereas the median remains remarkably stable at the level of $10^5$--$10^6$.
At the same time, the maximum observed value increases even more rapidly, reaching $10^{29}$ for $10^5$ trials. 
The pronounced discrepancy between the behavior of the mean and the median suggests that the sample mean is not representative of a typical realization but is instead dominated by rare extreme events.
As additional samples are collected, increasingly large realizations of $1/\sigma_r(\pA)^2$ continue to appear and exert a disproportionate influence on the empirical mean, while the central behavior captured by the median remains essentially unchanged. 

The above observation  suggests that the expectation $\mathbb{E}[1/\sigma^2_r(\pA)]$ may be infinite, i.e., the first moment does not exist. To further investigate this possibility, we examine the distributional properties of $1/\sigma^2_r(\pA)$ and its reciprocal $\sigma_r(\pA)$ using $10^5$ independent trials. The results are summarized in Fig.~\ref{fig:running_mean}--Fig.~\ref{fig:spacing}.

Fig.~\ref{fig:running_mean} plots the running  mean of $1/\sigma^2_r(\pA)$ as a function of the number of samples. Instead of converging to a finite limit, The running mean remains highly unstable and fails to exhibit clear convergence, which is a hallmark of a heavy‑tailed distribution with infinite mean.

\begin{figure}[!h]  
\vspace{-0.5cm}  
\setlength{\abovecaptionskip}{-0.2cm}   
\setlength{\belowcaptionskip}{0.2cm}   
\centering 
\begin{minipage}[b]{0.48\textwidth} 
\centering 
\includegraphics[width=1\textwidth]{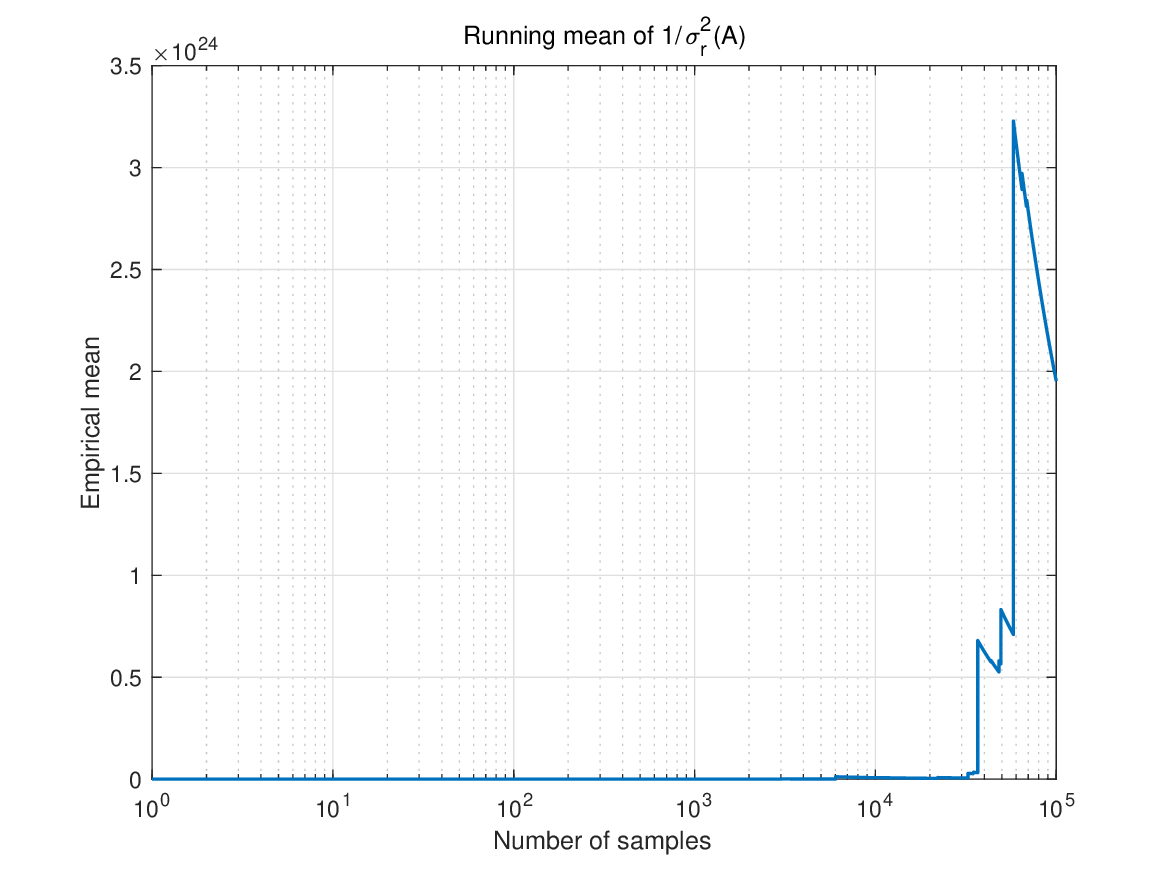} 
\vspace{-0.7cm}
\subcaption{Running mean of $1/\sigma^2_r(\pA)$ shows no convergence, suggesting infinite mean.\\ }
\label{fig:running_mean}
\end{minipage}
\vspace{-0.2cm}
\begin{minipage}[b]{0.48\textwidth} 
\centering 
\includegraphics[width=1\textwidth]{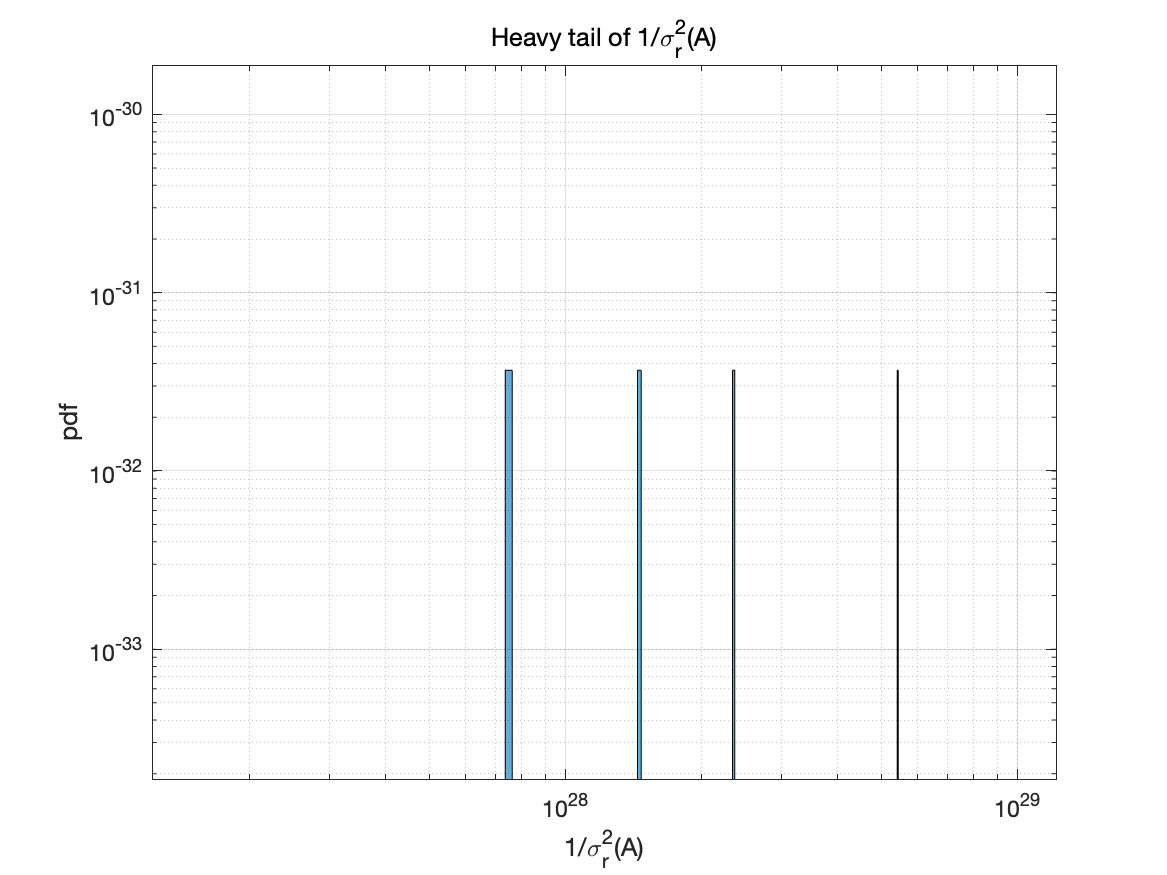} 
\vspace{-0.7cm}
\subcaption{Log–log pdf of $1/\sigma^2_r(\pA)$ reveals an extremely heavy tail spanning $>28$ orders of magnitude.}
\label{fig:hist_x}
\end{minipage}
\vspace{-0.2cm}
\begin{minipage}[b]{0.48\textwidth} 
\centering 
\includegraphics[width=1\textwidth]{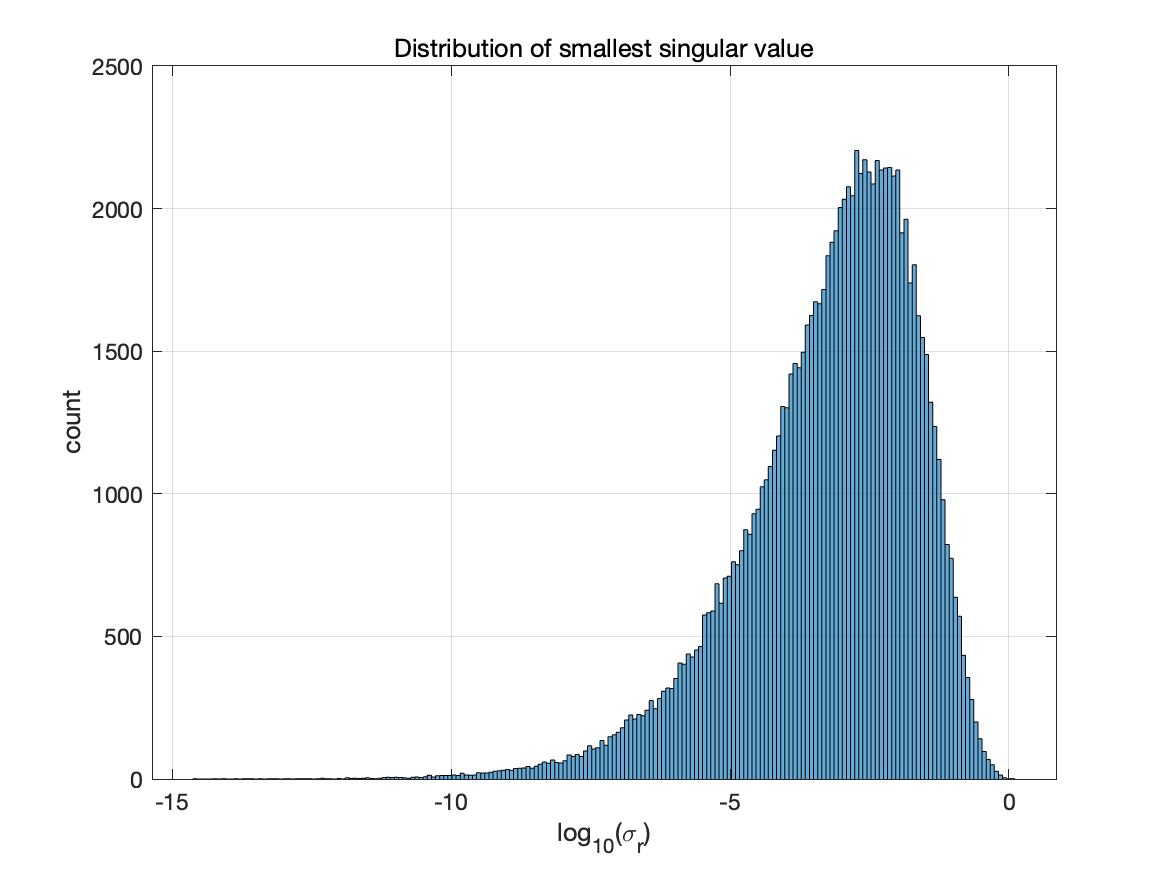} 
\vspace{-0.7cm}
\subcaption{Histogram of $\log_{10}\sigma_r(\pA)$ confirms $\sigma_r(\pA)$ can be as small as $\sim10^{-15}$. }
\label{fig:hist_sigma}
\end{minipage}
\begin{minipage}[b]{0.48\textwidth} 
\centering 
\includegraphics[width=1\textwidth]{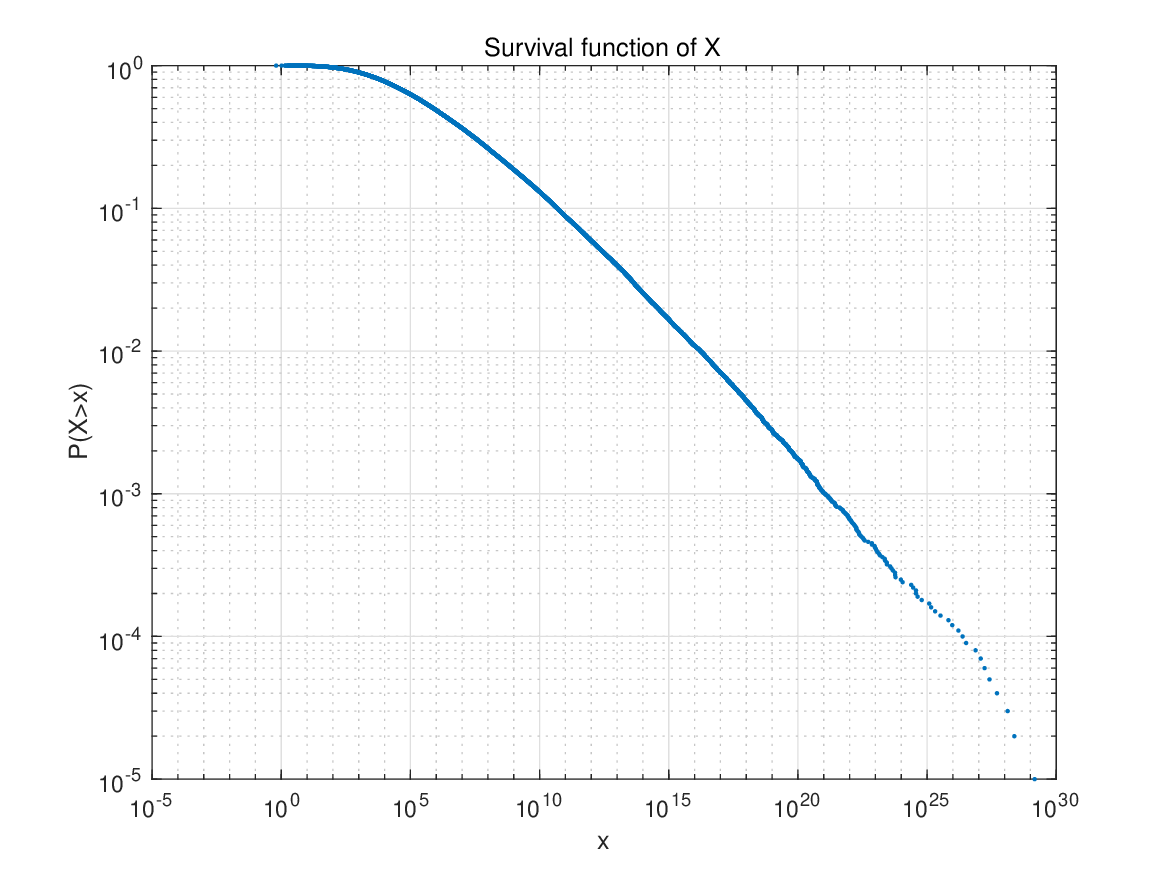} 
\vspace{-0.7cm}
\subcaption{Survival function $\mathbb{P}[X>x]$ with Hill estimate $\hat{\alpha}=0.2018<1$ implies $\mathbb{E}[X]=\infty$.}
\label{fig:survival}
\end{minipage}
\begin{minipage}[b]{0.48\textwidth} 
\centering 
\includegraphics[width=1\textwidth]{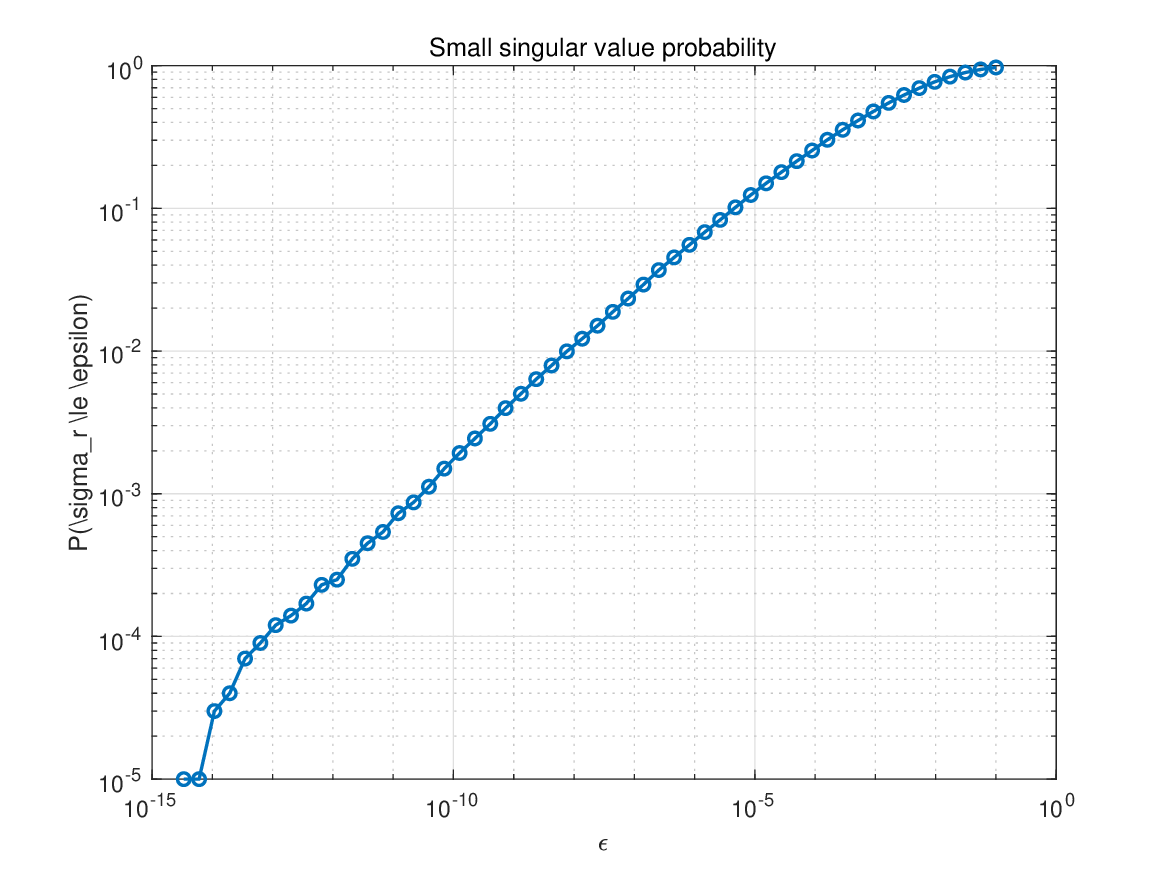} 
\vspace{-0.7cm}
\subcaption{Small singular value probability $\mathbb{P}[\sigma_r(\pA)\le\varepsilon]\sim\varepsilon^{\beta}$ with $\hat{\beta}=0.3669<2$ implies $\mathbb{E}[1/\sigma_r^2]=\infty$.}
\label{fig:small_sv}
\end{minipage}
\begin{minipage}[b]{0.48\textwidth} 
\centering 
\includegraphics[width=1\textwidth]{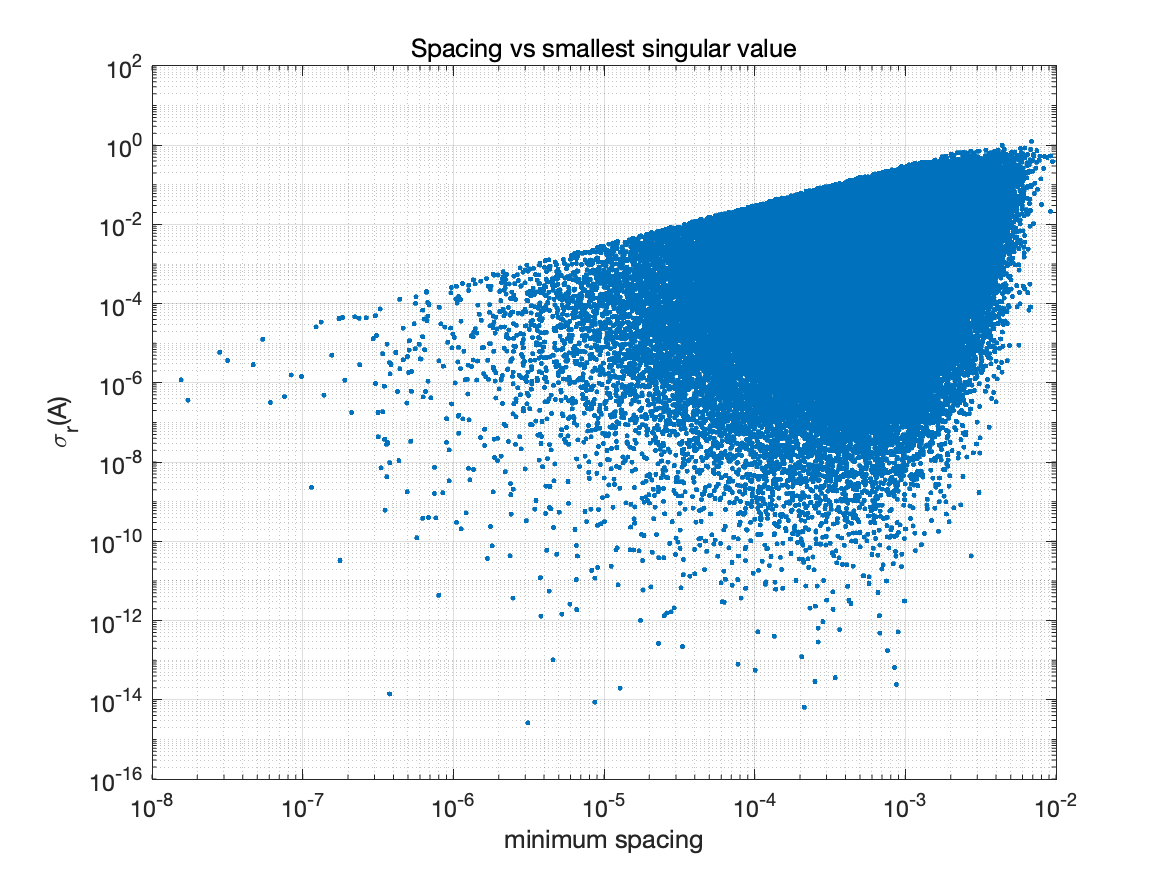} 
\vspace{-0.7cm}
\subcaption{Positive correlation between minimum spacing and $\sigma_r(\pA)$ reveals the geometric origin of extreme events.}
\label{fig:spacing}
\end{minipage}
\caption{Analysis of smallest singular value under the random uniform sampling nodes setting.}
\label{Fig:uniform}
\vspace{-0.6cm}
\end{figure}

Fig.~\ref{fig:hist_x} shows the probability density function (PDF) of $1/\sigma^2_r(\pA)$ in log--log scale. The support spans more than 28 orders of magnitude (up to $10^{28}$), and the density decays slowly over this enormous range, indicating an extremely heavy tail. In contrast, Fig.~\ref{fig:hist_sigma} displays the histogram of $\log_{10}\sigma_r(\pA)$, which is significantly left‑skewed, confirming that $\sigma_r(\pA)$ can take extremely small values (as low as $10^{-15}$), leading to the huge values observed for its reciprocal squared.

To quantify the tail decay, Fig.~\ref{fig:survival} presents the survival function $\mathbb{P}[X > x]$ with $X = 1/\sigma^2_r(\pA)$ in log--log coordinates. The curve is approximately linear, suggesting a power‑law tail $\mathbb{P}[X > x] \sim C x^{-\alpha}$. Applying the Hill estimator yields $\hat{\alpha} = 0.2018$ using the largest tail sample size $k=5000$ observations, which is far below 1. For a non‑negative random variable, $\mathbb{E}[X] < \infty$ if and only if $\alpha > 1$; thus $\hat{\alpha} < 1$ provides strong evidence that $\mathbb{E}[X]$ diverges.

A complementary perspective is given by the small singular value probability. Fig.~\ref{fig:small_sv} plots $\mathbb{P}[\sigma_r(\pA) \le \varepsilon]$ against $\varepsilon$ on a log--log scale. The relation is well approximated by a power law $\mathbb{P}[\sigma_r(\pA) \le \varepsilon] \sim C \varepsilon^{\beta}$ with $\hat{\beta}=0.3669$. 
Since
\[
\mathbb{E}\left[\frac{1}{\sigma_r^2}\right] = \int_0^\infty \mathbb{P}\left[\frac{1}{\sigma_r^2} > t\right] \mrd t = \int_0^\infty \mathbb{P}\left[\sigma_r \le \frac{1}{\sqrt{t}}\right] \mrd t,
\]
the expectation is finite if and only if $\beta > 2$. The estimated $\hat{\beta}=0.3669 < 2$ again indicates divergence.

Finally, Fig.~\ref{fig:spacing} provides evidence for a possible geometric mechanism underlying the observed extreme events. When two sampling nodes become exceptionally close, the corresponding rows of the Fourier matrix become nearly indistinguishable, which can substantially reduce the smallest singular value.

Taken together, these numerical observations provide compelling evidence that the distribution of $1/\sigma^2_r(\pA)$ possesses an extremely heavy tail and are consistent with the conjecture that $\mathbb{E}[1/\sigma^2_r(\pA)] = \infty$\footnote{It is noted that some extreme values approach the double‑precision limit; nevertheless, the conclusion still holds even when these samples are excluded.}. This finding motivates the subsequent theoretical investigation into the precise tail asymptotics and the (non‑)existence of moments for random Fourier matrices in the undersampled regime.

\subsection{Sample (b)}
We first set the noise to satisfy $\pe = (e_1, \dots, e_m)^\tp$ and $\text{Re}(\pe), \text{Im}(\pe) \sim \mathcal{N}(0,\frac{\epsilon^2}{2m}\pI_m)$ in the equidistant sampling model ($\delta_l=0$).
Fig.~\ref{Fig:equid_risk} illustrates the reconstruction risk as a function of the number of sampling points $m$ under the equidistant sampling. The red dots, which correspond to the theoretical values from Theorem \ref{Equidistant2}, closely match the empirical averages for all values of 
$m$.
The numerical results reveal two distinct regimes separated by a sharp transition near $m=N$. 
When $m<N$, the risk decreases slowly with increasing $m$, reflecting the persistent bias caused by the rank deficiency of the sampling matrix. 
As $m$ approaches $N$,  the bias vanishes and the risk exhibits a sudden drop. 
Beyond the transition ($m\ge N$), the reconstruction risk enters a different regime and decays much more smoothly.
In this region, the decay follows the theoretical prediction
\[
\mathbb{E}\|\hat{\px}-\px\|^2 =\frac{N\epsilon^2}{m^2},
\]
showing an algebraic rate of order $m^{-2}$.
The excellent agreement between the numerical results and the theoretical values confirms the sharpness of the analysis in Theorem~\ref{Equidistant2}.

\begin{figure}[!t] 
\vspace{-0.3cm}  
\setlength{\abovecaptionskip}{-0.1cm}   
\setlength{\belowcaptionskip}{0.1cm}   
\centering 
\begin{minipage}[b]{0.45\textwidth} 
\centering 
\includegraphics[width=1\textwidth]{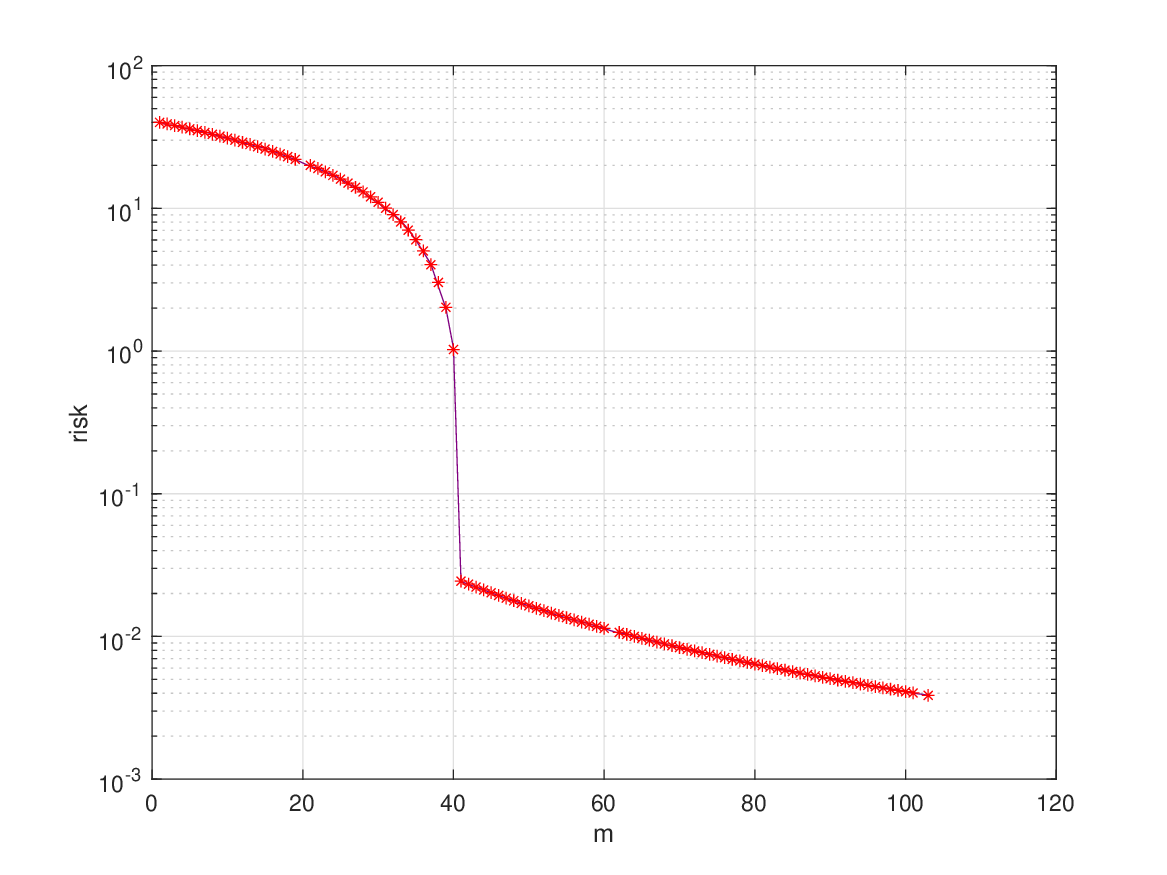} 
\vspace{-0.7cm}
\subcaption{$N=41$.}
\end{minipage}
\begin{minipage}[b]{0.45\textwidth} 
\centering 
\includegraphics[width=1\textwidth]{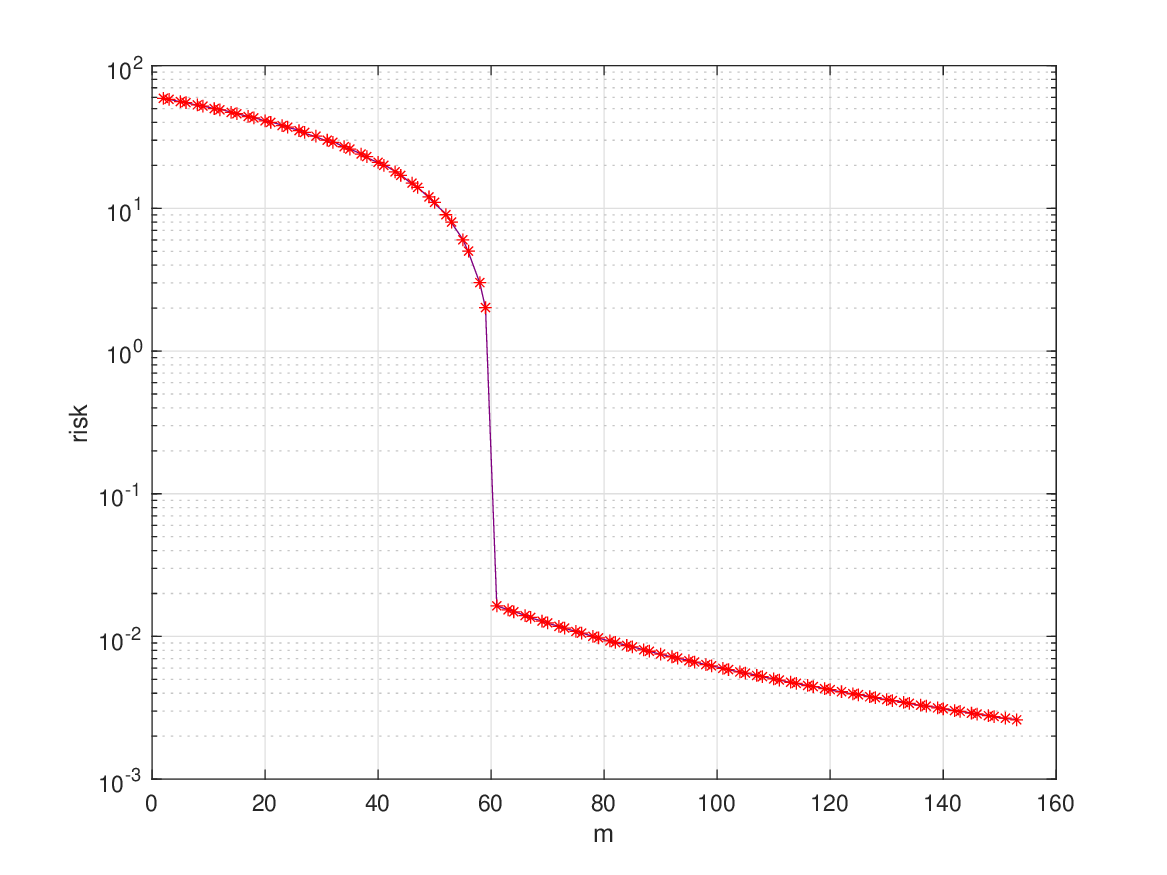} 
\vspace{-0.7cm}
\subcaption{$N=61$.}
\end{minipage}
\caption{Plot of the risk as a function of $m$ in the equidistant sampling model $t_l = \frac{l-1}{m}, l=1,\dots m$. Here, the noise $\pe = (e_1, \dots, e_m)^\tp$ satisfies $\text{Re}(\pe), \text{Im}(\pe) \sim \mathcal{N}(0,\frac{\epsilon^2}{2m}\pI_m)$.
For each  $m$, we generated $1000$ independent realizations of the noise $\pe$, computed the least-squares estimator $\hat{\px}$, and plotted the empirical average of $\|\hat{\px}-\px\|^2$.
The red dots represent numerical evaluations of the theoretical expression in Theorem \ref{Equidistant2}.}
\label{Fig:equid_risk}
\end{figure}

In Fig.~\ref{Fig:random_risk_ud}, we present numerical experiments for the sampling model (b) with small random perturbations in the sampling nodes. The perturbations $\delta_l$ are independently generated according to the prescribed random model. The results show that the sharp phase-transition phenomenon persists in the presence of random perturbations. In particular, a significant change in the reconstruction risk is observed around the critical sampling level $m=N$, where the system transitions from the underdetermined regime ($m<N$) to the overdetermined regime ($m \geq N$).
The blue and red dots represent the lower and upper bounds derived from Theorem \ref{thm:random_smallT_explicit}, respectively. We observe that the lower bound provides a relatively accurate characterization of the reconstruction error over the entire range of $m$. In contrast, the upper bound becomes considerably loose when $m<N$ and $m$ approaches $N$, leading to excessively large estimates. Improving this upper bound, especially near the critical transition region, remains an important direction for future work.
 
\begin{figure}[!t] 
\vspace{-0.3cm}  
\setlength{\abovecaptionskip}{-0.1cm}   
\setlength{\belowcaptionskip}{0.1cm}   
\centering 
\begin{minipage}[b]{0.4\textwidth} 
\centering 
\includegraphics[width=1\textwidth]{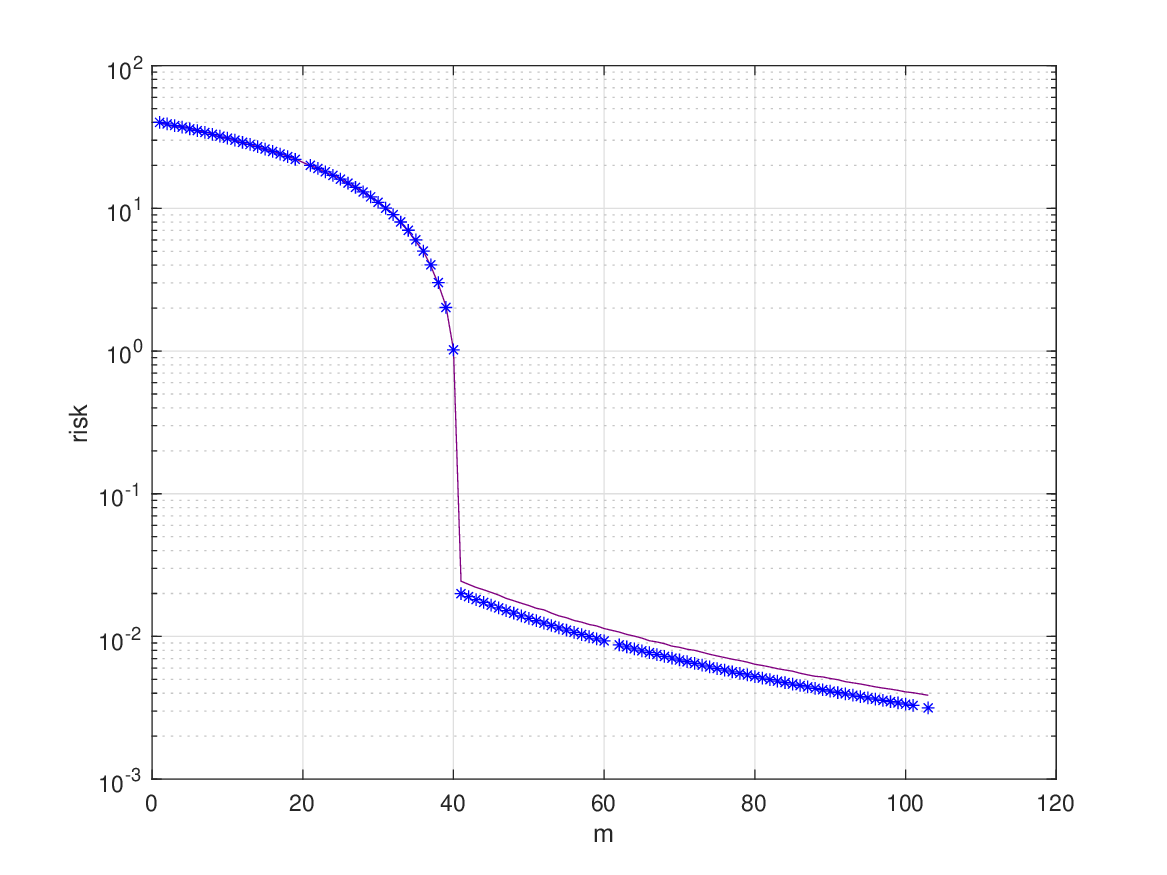} 
\vspace{-0.8cm}
\subcaption{$N=41$.}
\end{minipage}
\begin{minipage}[b]{0.4\textwidth} 
\centering 
\includegraphics[width=1\textwidth]{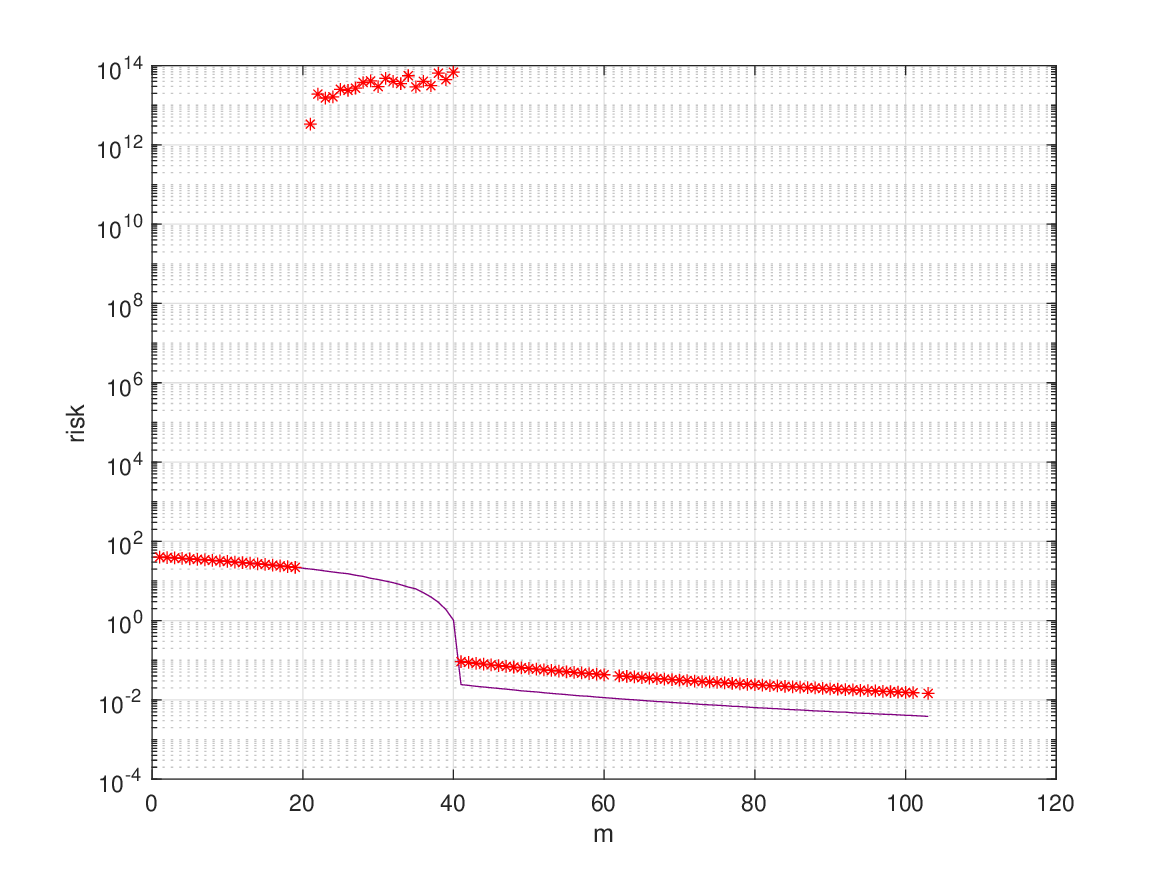} 
\vspace{-0.8cm}
\subcaption{$N=41$.}
\end{minipage}
\begin{minipage}[b]{0.4\textwidth} 
\centering 
\includegraphics[width=1\textwidth]{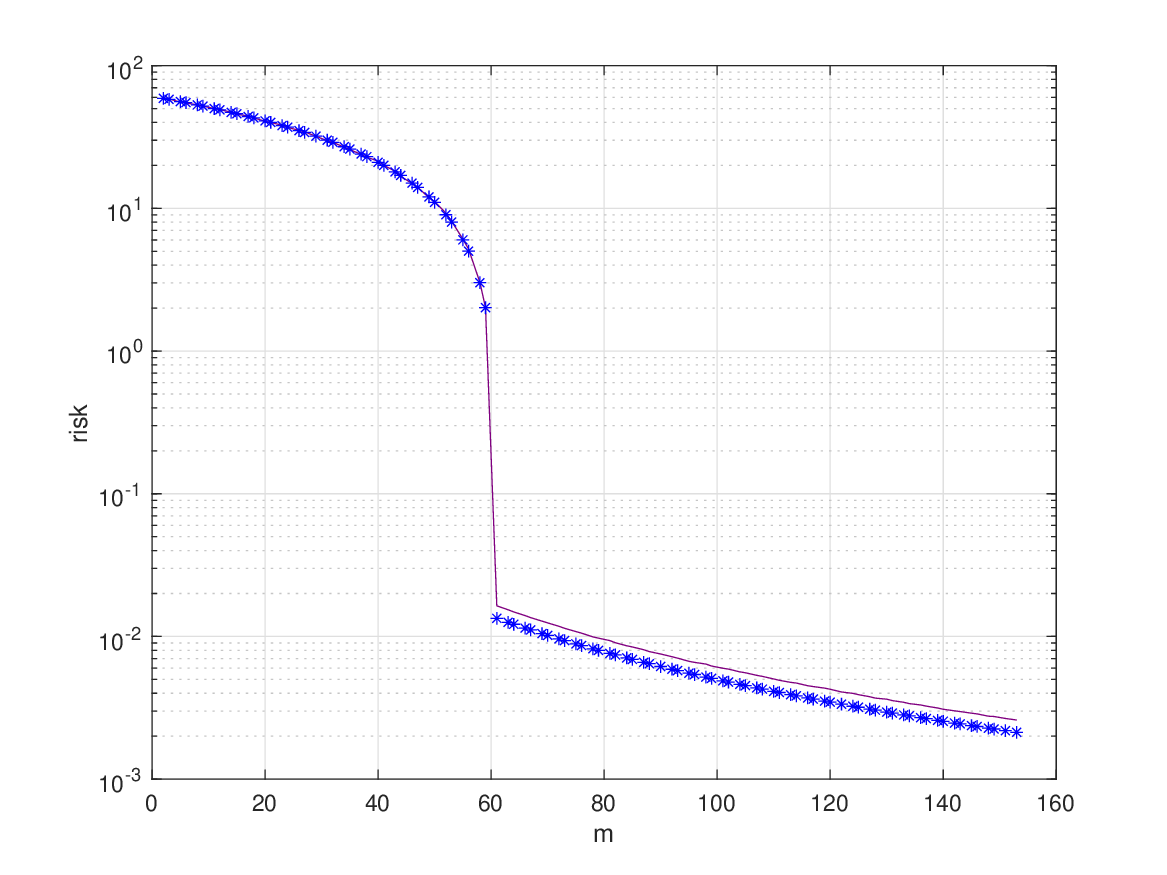} 
\vspace{-0.7cm}
\subcaption{$N=61$.}
\end{minipage}
\begin{minipage}[b]{0.4\textwidth} 
\centering 
\includegraphics[width=1\textwidth]{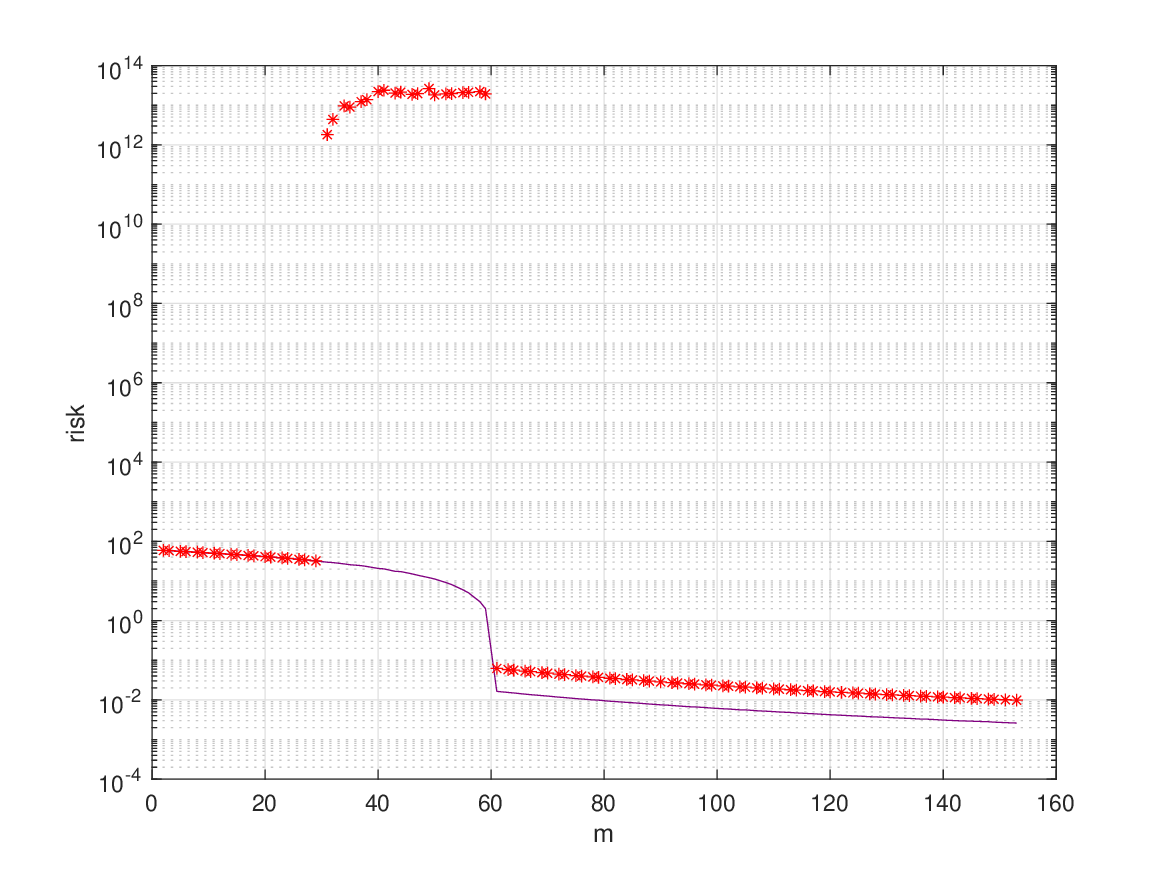} 
\vspace{-0.7cm}
\subcaption{$N=61$.}
\end{minipage}
\caption{
Plot of the risk as a function of $m$ in the Sample model (b) with small random perturbation. The blue and red dots denote the lower and upper bounds, respectively, obtained from Theorem \ref{thm:random_smallT_explicit}.}
\label{Fig:random_risk_ud}
\end{figure}

Following the numerical investigation of the distributional properties of the smallest singular value under randomly and uniformly distributed sampling nodes in Section \ref{sec:num1}, we further perform numerical experiments for two other sampling models: deterministic equidistant sampling and random non-equidistant sampling. The results reveal a significant difference in the spectral behavior of the smallest singular value between these two settings.

Fig.~\ref{Fig:equid1} shows the running mean of $1/\sigma^2_r(\pA)$. Unlike the random sampling case, the empirical mean remains essentially constant throughout the entire sampling process, indicating that no heavy-tail behavior is observed. This is due to the deterministic nature of the sampling matrix, where the smallest singular value remains nearly unchanged across trials.
Fig.~\ref{Fig:equid2} presents the empirical distribution of $1/\sigma^2_r(\pA)$ in log--log scale. The distribution is highly concentrated around a single value, suggesting that the reciprocal squared smallest singular value does not exhibit significant variability. Correspondingly, Fig.~\ref{Fig:equid3} shows that $\log_{10}\sigma_r(\pA)$ is concentrated at a fixed location, confirming the absence of extremely small singular values. Therefore, the instability caused by rare ill-conditioned realizations, which is observed in the random sampling setting, does not occur for deterministic equidistant sampling.

\begin{figure}[!t] 
\vspace{-0.3cm}  
\setlength{\abovecaptionskip}{-0.1cm}   
\setlength{\belowcaptionskip}{0.1cm}   
\centering 
\begin{minipage}[b]{0.45\textwidth} 
\centering 
\includegraphics[width=1\textwidth]{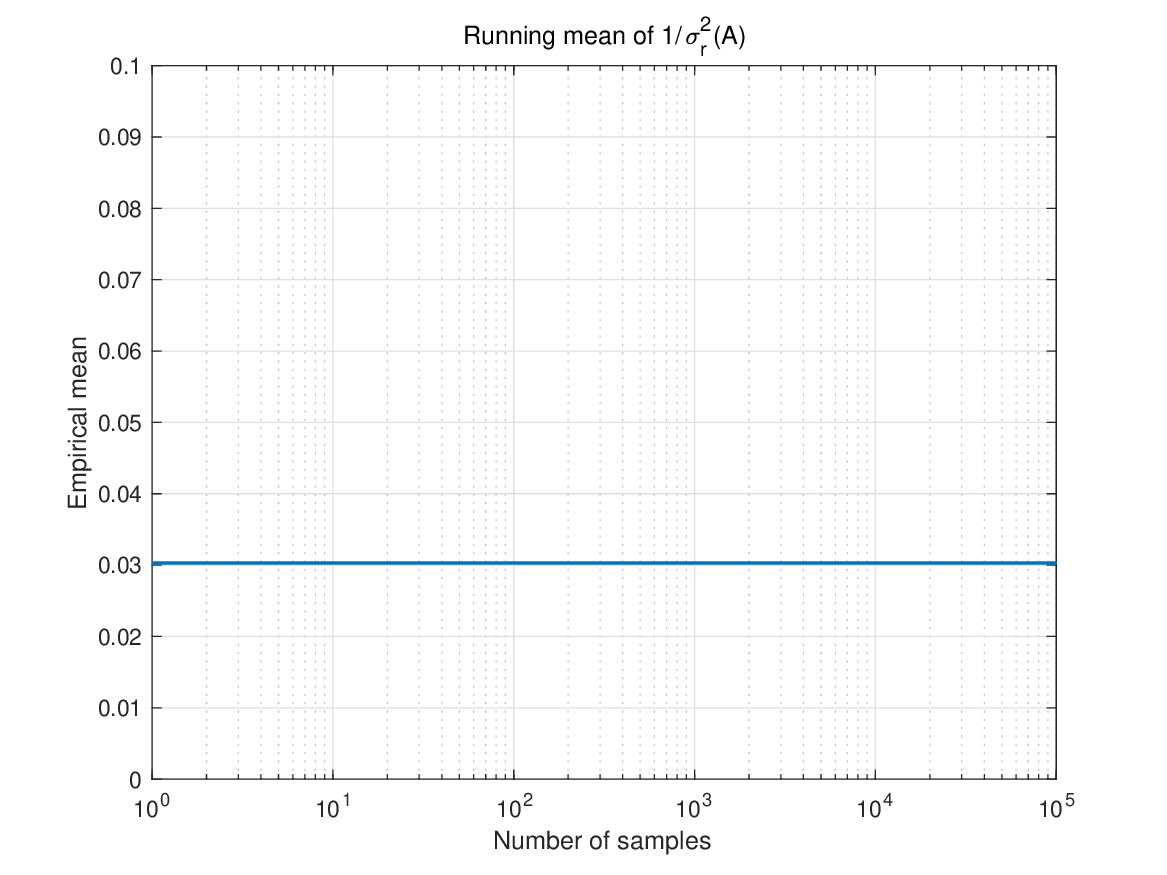} 
\vspace{-0.7cm}
\subcaption{Running mean of $1/\sigma^2_r(\pA)$.}
\label{Fig:equid1}
\end{minipage}
\begin{minipage}[b]{0.45\textwidth} 
\centering 
\includegraphics[width=1\textwidth]{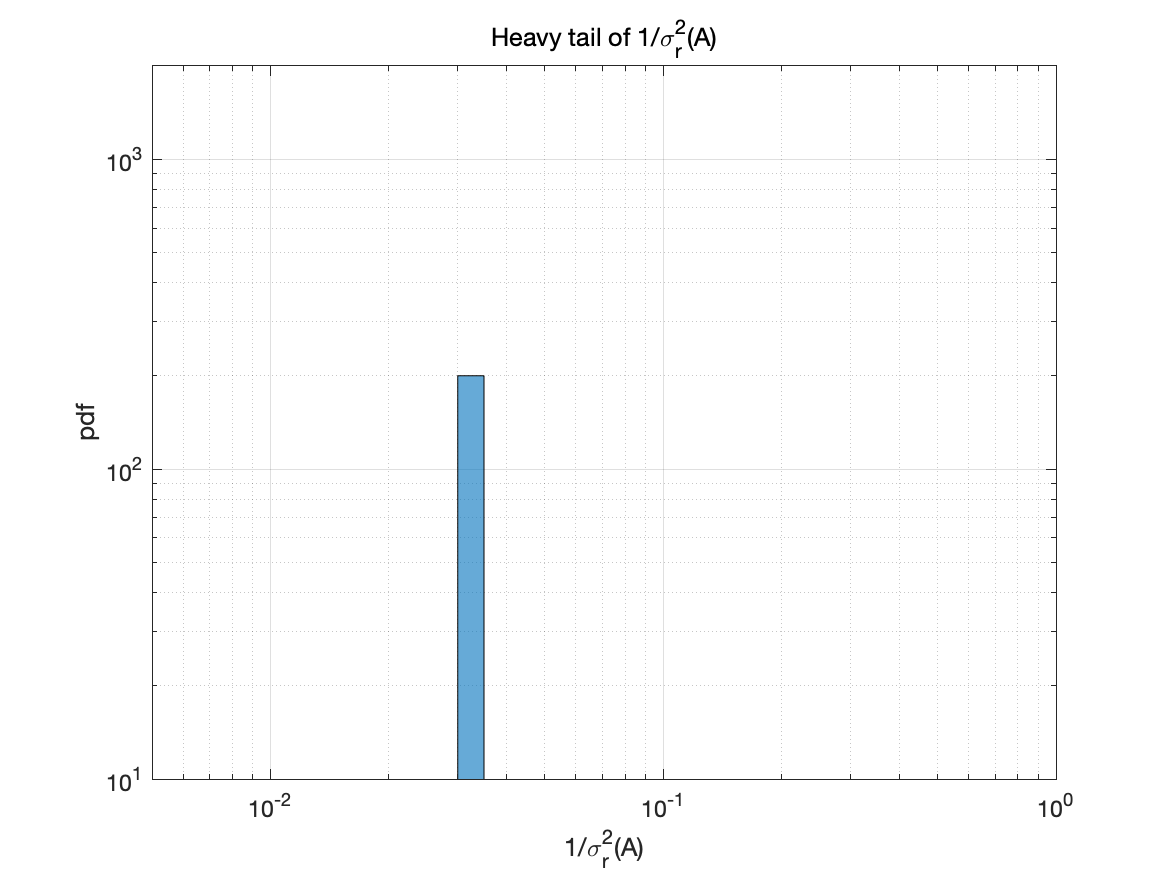} 
\vspace{-0.7cm}
\subcaption{Log–log pdf of $1/\sigma^2_r(\pA)$.}
\label{Fig:equid2}
\end{minipage}
\begin{minipage}[b]{0.45\textwidth} 
\centering 
\includegraphics[width=1\textwidth]{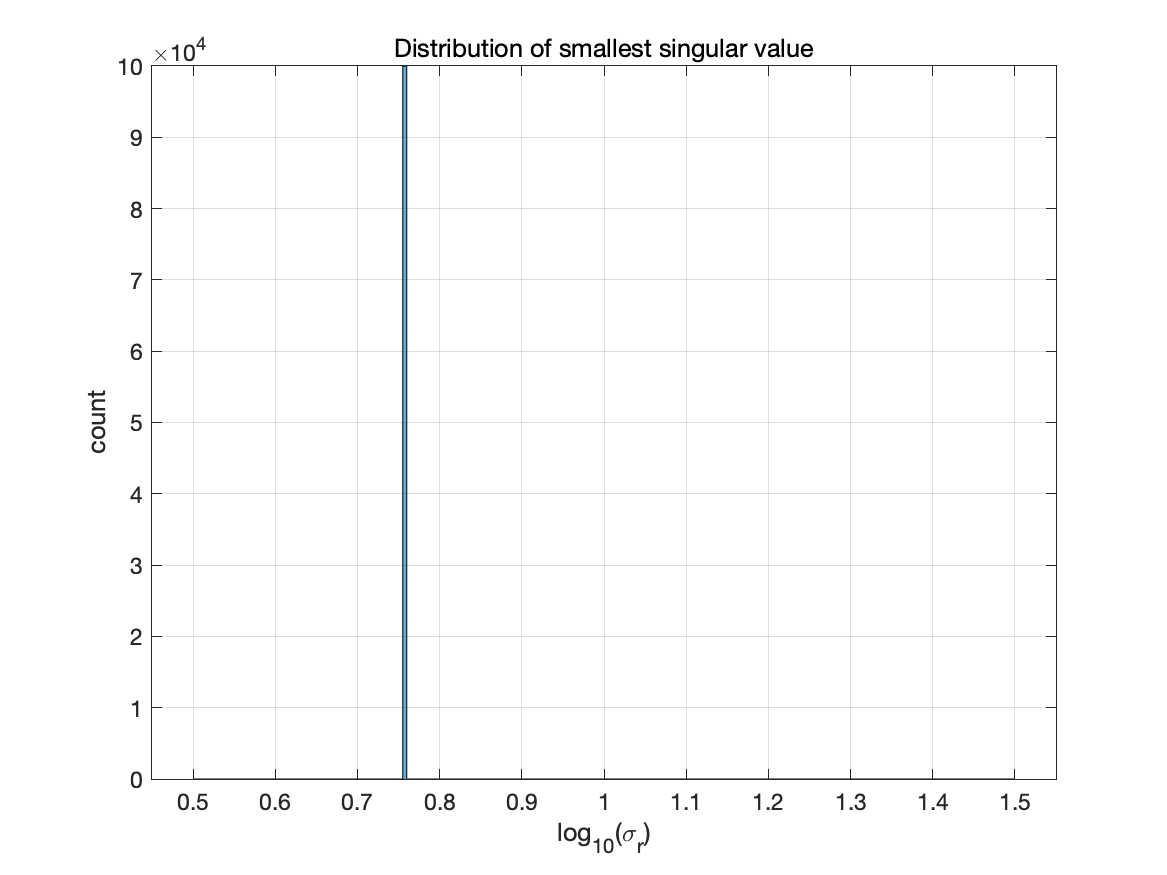} 
\vspace{-0.7cm}
\subcaption{Histogram of $\log_{10}\sigma_r(\pA)$.}
\label{Fig:equid3}
\end{minipage}
\caption{Analysis of smallest singular value under the deterministic equidistant sampling nodes setting.}
\label{Fig:equid}
\vspace{-0.3cm}
\end{figure}

Fig.~\ref{Fig:random1} illustrates the running mean of $1/\sigma^2_r(\pA)$. Although the empirical mean appears to stabilize for the current number of samples, the convergence behavior is much slower and the fluctuations indicate the presence of a nontrivial tail distribution. This suggests that rare sampling configurations with very small singular values may significantly affect the expectation.
Fig.~\ref{Fig:random2} displays the empirical density of $1/\sigma^2_r(\pA)$. Compared with the deterministic case, the distribution spreads over a much wider range, revealing substantial variability caused by random perturbations of the sampling nodes. Fig.~\ref{Fig:random3} further shows the distribution of $\log_10\sigma_r(\pA)$, which has a non-negligible spread and allows $\sigma_r(\pA)$ to attain relatively small values. These small singular-value events lead to large values of $1/\sigma^2_r(\pA)$, explaining the increased instability of the reconstruction process under random sampling.

\begin{figure}[!t] 
\vspace{-0.3cm}  
\setlength{\abovecaptionskip}{-0.1cm}   
\setlength{\belowcaptionskip}{0.1cm}   
\centering 
\begin{minipage}[b]{0.45\textwidth} 
\centering 
\includegraphics[width=1\textwidth]{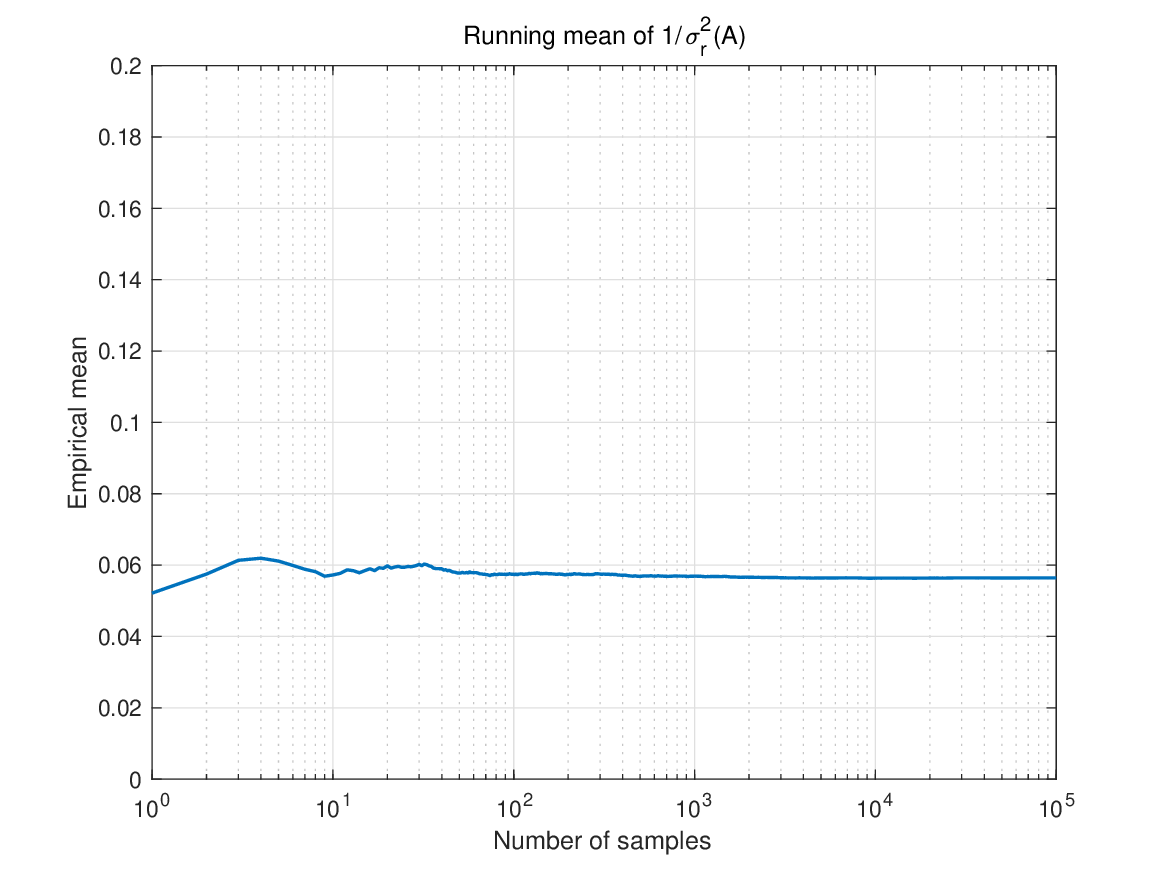} 
\vspace{-0.7cm}
\subcaption{Running mean of $1/\sigma^2_r(\pA)$.}
\label{Fig:random1}
\end{minipage}
\begin{minipage}[b]{0.45\textwidth} 
\centering 
\includegraphics[width=1\textwidth]{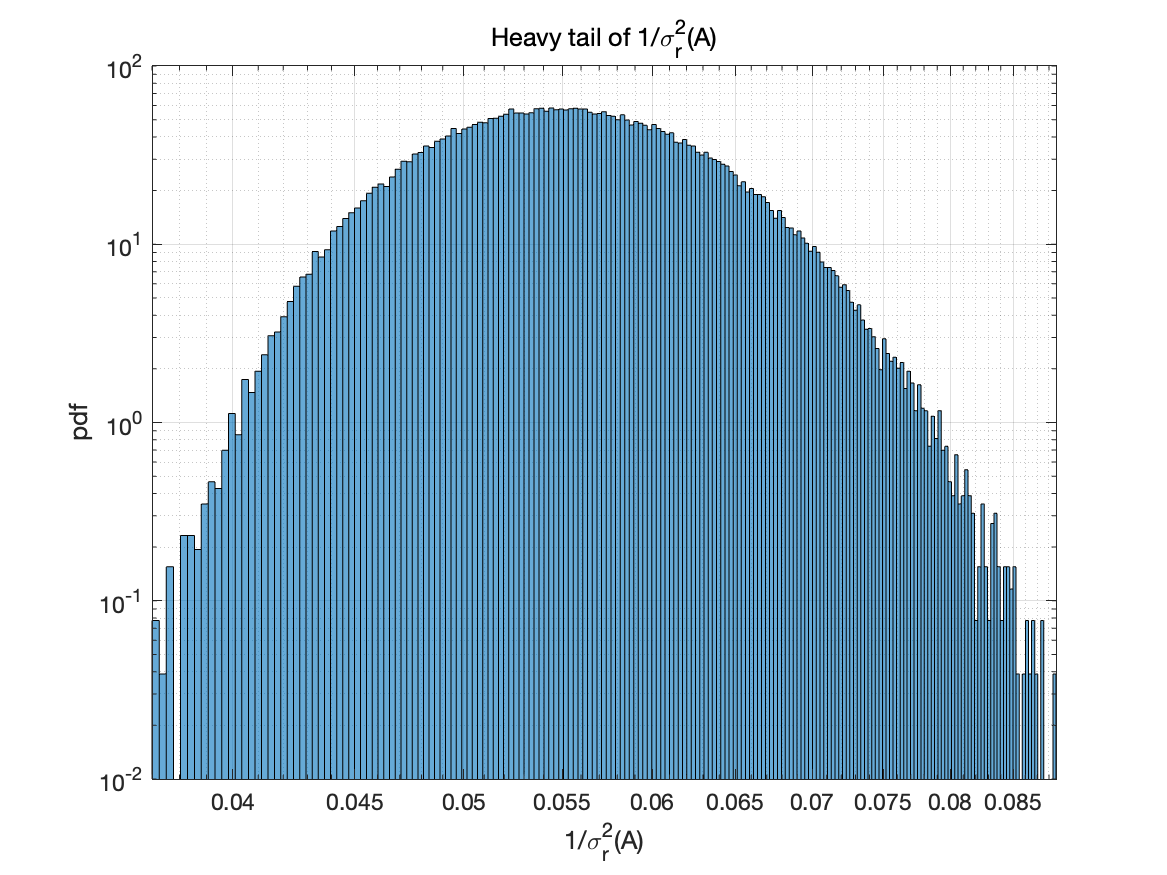} 
\vspace{-0.7cm}
\subcaption{Log–log pdf of $1/\sigma^2_r(\pA)$.}
\label{Fig:random2}
\end{minipage}
\begin{minipage}[b]{0.45\textwidth} 
\centering 
\includegraphics[width=1\textwidth]{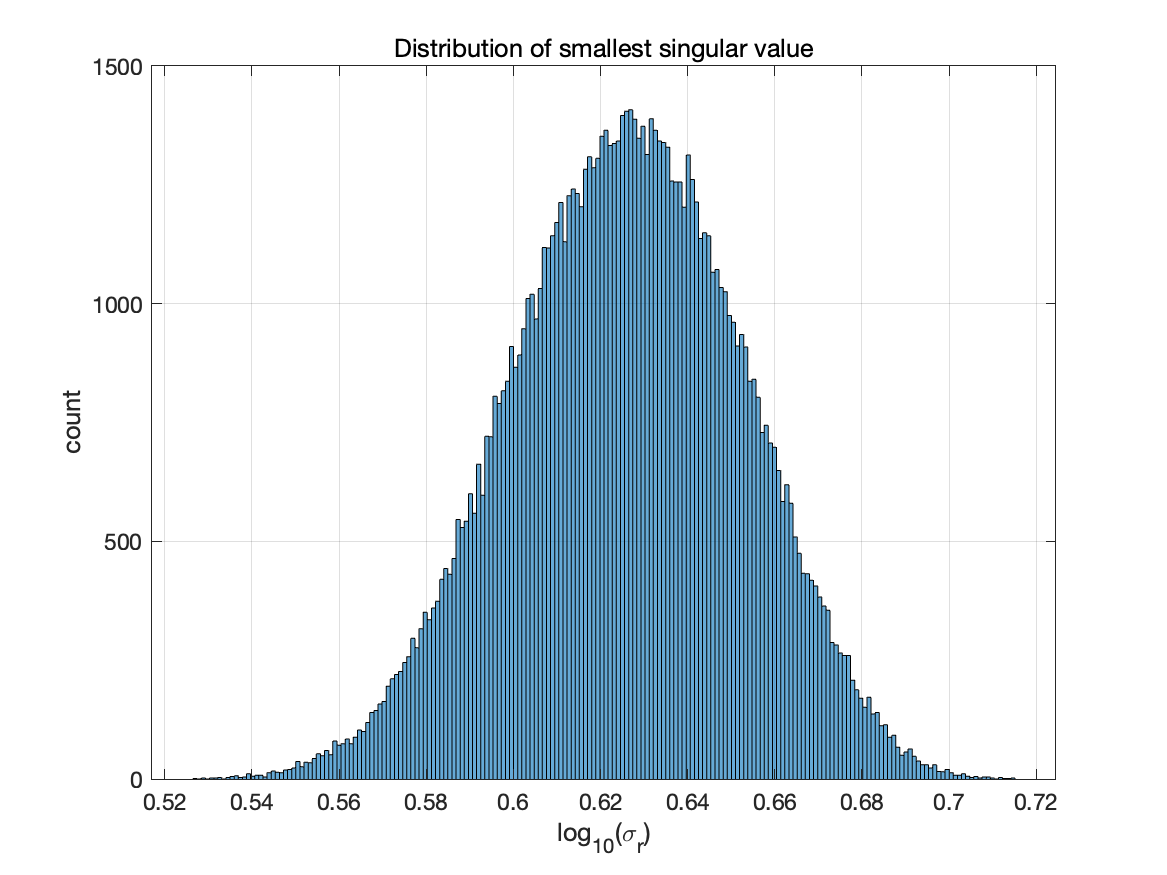} 
\vspace{-0.7cm}
\subcaption{Histogram of $\log_{10}\sigma_r(\pA)$.}
\label{Fig:random3}
\end{minipage}
\caption{Analysis of smallest singular value under the random non-equidistant sampling nodes setting.}
\label{Fig:random}
\end{figure}

\section{Conclusions} \label{sec:conclusions}
In this work, we investigate two sampling models for the random trigonometric polynomial considered in \cite{bass2005random} and provide a theoretical explanation for the distinct reconstruction behaviors observed in numerical experiments. Based on the non-asymptotic expressions for the reconstruction risk, we characterize how the sampling scheme influences the statistical performance of least-squares reconstruction.
For independent uniform sampling, where the sampling nodes $t_l$are drawn i.i.d. from the uniform distribution on $[0,1)$, we prove that the expected reconstruction risk diverges. This result provides a theoretical explanation for the instability observed in numerical experiments, where rare nearly singular sampling configurations dominate the reconstruction error. For equidistant sampling, we derive an explicit expression for the expected reconstruction risk, which agrees with the numerical observations and provides a benchmark for structured sampling schemes.
Furthermore, we investigate jittered sampling, where the sampling nodes are generated as random perturbations of a uniform grid. We establish upper and lower bounds for the expected reconstruction risk in this setting. While the obtained lower bound captures the empirical behavior reasonably well, the upper bound becomes loose in the regime $m<N$ with $m$ close to $N$. This limitation arises from the use of Bernstein-type concentration inequalities in our analysis, which generally do not provide a sufficiently precise characterization of the location of spectral transitions.
Improving these bounds remains an interesting direction for future research. We believe that a more refined analysis of the Fourier--Vandermonde structure, particularly the associated Gram matrix $\pA\pA^*$, may provide a promising approach. Since this matrix can be viewed as a randomly perturbed Toeplitz-type matrix, a direct investigation of its smallest eigenvalue and its concentration behavior may lead to sharper characterizations of the reconstruction risk and a deeper understanding of the spectral mechanisms underlying interpolation transitions.


\section*{Declarations}
The authors declare no competing interests.

\begin{appendices}

\section{Proof of Main Theorems} \label{sec:3}
In this section, we provide the proofs of the main theorems.
\subsection{Proof of Theorem \ref{thm1}}
\begin{proof}[Proof of Theorem \ref{thm1}]
We can get
\begin{align*}
\| \hx - \px \|^2 =& \| \pA^{\dagger}\pA\px -\px +\pA^{\dagger}\pe \|^2  \\
=& \| \pA^{\dagger}\pe \|^2 + \| (\pA^{\dagger}\pA-\pI_N)\px \|^2 + 2 \langle \pA^{\dagger}\pe,  (\pA^{\dagger}\pA-\pI_N)\px\rangle.
\end{align*}
Since $\text{Re}(\pe), \text{Im}(\pe) \sim \mathcal{N}(0,\frac{\epsilon^2}{2m}\pI_m)$, we get
\begin{align}\label{t0}
\mbE[\pe\pe^*] = \frac{\epsilon^2}{m} \pI_m.
\end{align}
From  \eqref{xx} and \eqref{t0}, for given $\{t_l\} \in [0,1]^m$ in $A_{l,k}$,
we then obtain
\begin{align}\label{t1}
\mbE[\| \hx - \px \|^2 | t_l] = \frac{\epsilon^2}{m} \| \pA^{\dagger}\|_F^2 +  \| \pA^{\dagger}\pA-\pI_N \|_F^2,
\end{align}
because $\pe$ is independent of $\px$. 
Assume that the SVD of $\pA$ is $\pA = \pU\pSig\pV^*$ with $\pU \in \mbC^{m \times r}$, $\pV \in \mbC^{N \times r}$ and $\pSig \in \mbC^{r \times r}$.
Based on \eqref{A_dagger}, we first compute that 
\begin{align}\label{t2}
\| \pA^{\dagger}\|_F^2 = \| \pV\pSig^{-1}\pU^* \|_F^2 = \| \pSig^{-1}\|_F^2 = \sum_{i=1}^r \frac{1}{\sigma_i^2(\pA)},
\end{align}
where $r=\rank(\pA)$.
Based on \eqref{A_dagger}, we also compute again that
\begin{align}
\| \pA^{\dagger}\pA-\pI_N \|_F^2 = \| \pI_N -\pV\pV^*\|_F^2.
\end{align}
Since $\pV^*\pV=\pI_{r}$, there exists a matrix $\pV_{\perp} \in \mbC^{N\times (N-r)}$ satisfying $\pV_{\perp}\pV_{\perp}^*+ \pV\pV^* = \pI_N$ and $\pV_{\perp}^*\pV_{\perp} = \pI_{N-r}$.
Then we have
\begin{align}\label{t3}
\| \pA^{\dagger}\pA-\pI_N \|_F^2 = \| \pV_{\perp}\pV_{\perp}^* \|_F^2 =  \langle \pV_{\perp}^*\pV_{\perp},\pV_{\perp}^*\pV_{\perp} \rangle = N-r.
\end{align}
Combining with \eqref{t1}, \eqref{t2} and \eqref{t3}, we can get 
\begin{align}
\mbE[\| \hx - \px \|^2 | t_l] = \frac{\epsilon^2}{m}  \sum_{i=1}^r \frac{1}{\sigma_i^2(\pA)} +N-r.
\end{align}
Then, we have
\begin{align}
\mbE[\| \hx - \px \|^2] =  \frac{\epsilon^2}{m} \mbE \left[ \sum_{i=1}^r \frac{1}{\sigma_i^2(\pA)} \right] + \mbE \left[N-r \right].
\end{align}
\end{proof}

\subsection{Proof of Theorem \ref{Thm:Divergence}}
\begin{proof}[Proof of Theorem \ref{Thm:Divergence}]
Define the smallest nonzero separation
\begin{align}  \label{def:delta}
\Delta^+_{\min} = \min_{i\ne j, t_i\ne t_j}  \dist_T(t_{i},t_j).
\end{align}
And we know that $r = \min\{m,N\}$ almost surely from Proposition \ref{pro4}.

\paragraph{Step 1: Upper bound on the smallest singular value}
By Proposition~\ref{pro:sigmin},  we have
\[
\sigma_{r}(\pA) = \min_{\|\px\|_2=1} \|\pA\px\|_2
\le \frac{\|\pA\pz\|_2}{\|\pz\|_2}  \lesssim N^{3/2}\Delta^+_{\min}.
\]
Therefore, we can get
\begin{align}  \label{eq:lowerbound}
\frac1{\sigma_{r}^2(\pA)}
\gtrsim
\frac1{N^3(\Delta^+_{\min})^2}.
\end{align}

\paragraph{Step 2: Minimal spacing}

\begin{lemma}
\label{lem:minspacing}
Let  $t_1,\dots,t_m $ be i.i.d. random variables uniformly distributed on  $[0,1) $, and assume that  $\varepsilon \le O(\frac{1}{m^2})$. 
Then there exist constants $c,C>0$ such that for all sufficiently small  $\varepsilon>0 $,
\[
c m^2\varepsilon \le \mathbb{P}[\Delta_{\min}\le\varepsilon] \le C m^2\varepsilon.
\]
Particularly, 
\[
c m^2\varepsilon \le \mathbb{P}[\Delta^+_{\min}\le\varepsilon] \le C m^2\varepsilon
\]
also holds for $c,C>0$ and sufficiently small positive constant $\varepsilon \le O(\frac{1}{m^2})$.
\end{lemma}

\begin{proof}
For any fixed pair $(x,y)$ and $\epsilon \in (0,1)$, we can obtain formula 
\[
\mathbb{P}[ |x-y| \leq \varepsilon]=1-(1-\varepsilon)^2 = 2\varepsilon-\varepsilon^2 \leq 2\varepsilon
\]
by calculating the area of the shaded region (see Fig.~\ref{Fig:area_method}). 
\begin{figure}[htbp]
\vspace{-0.4cm}  
\setlength{\abovecaptionskip}{-0.1cm}   
\setlength{\belowcaptionskip}{-0.2cm}   
  \centering
\includegraphics[width=0.45\textwidth]{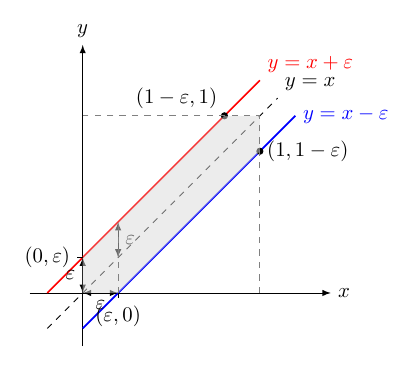}
  \vspace{-0.2cm}
  \caption{Area of the shaded region represents the probability of $|x-y| \leq \varepsilon$.}
  \label{Fig:area_method}
\end{figure}

Define $E_{ij}:=\{|t_i-t_j|\le\varepsilon\}$. 
Then, we know
\vspace{-0.2cm}
\[
 \{\Delta_{\min}\le\varepsilon\} = \bigcup_{i\ne j}E_{ij}.
\vspace{-0.2cm}
\]
By a union bound,
\begin{align*}
 \mathbb{P}[\Delta_{\min}\le\varepsilon] =& \mathbb{P}\Big[\bigcup_{i\ne j} E_{ij} \Big] \leq \sum_{i\ne j}\mathbb{P} [E_{ij}] \\
\leq& \frac{m(m-1)}{2}\mathbb{P}[|{x}-y|\le\varepsilon]
\le m(m-1) \varepsilon \le m^2 \varepsilon.
\end{align*}

For the lower bound, define $X:=\sum_{i\ne j}\mathbf{1}_{E_{ij}}$. 
Then, we know
\[
\{\Delta_{\min}\le\varepsilon\} = \{X\ge1\}.
\]
Moreover,
\[
 \mathbb{E}[X] = \sum_{i\ne j}\mathbb{P}[E_{ij}]= \frac{m(m-1)}{2} (2\varepsilon-\varepsilon^2)
\asymp m^2\varepsilon.
\]
We next estimate the second moment:
\[
X^2 = \sum_{i\ne j}\mathbf{1}_{E_{ij}} + 2\sum_{(i\ne j)\ne(k\ne \ell)} \mathbf{1}_{E_{ij}}\mathbf{1}_{E_{k\ell}}.
\]
Hence, 
\[
\mathbb{E}[X^2] = \mathbb{E}[X] + 2\sum_{(i\ne j)\ne(k\ne \ell)}
\mathbb{P}[E_{ij}\cap E_{k\ell}].
\]
If the pairs \((i,j)\) and \((k,\ell)\) are disjoint, then the events are independent and
\[
\mathbb{P}[E_{ij}\cap E_{k\ell}] = \mathbb{P}[E_{ij}]\mathbb{P}[E_{k\ell}] \lesssim \varepsilon^2.
\]
There are $O(m^4)$ such terms. If the pairs are overlapping, for example
\[
E_{12}\cap E_{13},
\]
which means $|t_1-t_2|\leq \varepsilon$ and $|t_1-t_3|\leq \varepsilon$,
then fixing \(t_1\), both \(t_2\) and \(t_3\) must lie in an interval of length \(2\varepsilon\), thus
\[
\mathbb{P}[E_{12}\cap E_{13}] \lesssim \varepsilon^2.
\]
There are $O(m^3)$ such terms. 
Consequently, we can get
\[
\mathbb{E}[X^2] \lesssim m^2\varepsilon + m^4\varepsilon^2.
\]
Assume now that $\varepsilon \le O(\frac{1}{m^2})$, then $m^4\varepsilon^2 \lesssim m^2\varepsilon$,
and therefore
\[
\mathbb{E}[X^2] \lesssim m^2\varepsilon.
\]
Applying the Paley--Zygmund inequality, we can get
\[
\mathbb{P}[X\ge1] \ge \frac{(\mathbb{E}[X])^2}{\mathbb{E}[X^2]}
\gtrsim m^2\varepsilon.
\]
Since $\mathbb{P}[X\ge1] = \mathbb{P}[\Delta_{\min}\le\varepsilon]$, 
we then have
\[
c m^2\varepsilon \le \mathbb{P}[\Delta_{\min}\le\varepsilon] \le C m^2\varepsilon.
\]

Since $t_1,\dots,t_m$ are continuously distributed,
\[
\mathbb P[t_i=t_j]=0,\qquad i\neq j,
\]
and hence, $\Delta_{\min}=\Delta_{\min}^+$ holds almost surely.
Therefore, the same estimate holds for $\Delta_{\min}^+$.
\end{proof}

\paragraph{Step 3: Divergence of the inverse moment}

Combining \eqref{eq:lowerbound} and Lemma~\ref{lem:minspacing}, we obtain
\[
\mathbb{E}\left[\frac{1}{\sigma_{r}^2(\pA)}\right] \gtrsim \frac{1}{N^3} \mathbb{E}\left[\frac{1}{(\Delta^+_{\min})^2}\right].
\]
We now estimate the inverse moment of $\Delta^+_{\min}$. For any nonnegative random variable $Y$,
\[
\mathbb{E}[Y] = \int_0^\infty \mathbb{P}[Y>t] \mrd t.
\]
Applying this identity to $Y=(\Delta_{\min}^+)^{-2}$, we get
\[
\mathbb{E} \left[\frac1{(\Delta_{\min}^+)^2}\right] = \int_0^\infty \mathbb{P}\left[(\Delta_{\min}^+)^{-2}>t\right]\mrd t.
\]
Since
\[
\{(\Delta_{\min}^+)^{-2}>t\} = \{\Delta^+_{\min}<t^{-1/2}\},
\]
the change of variables $\varepsilon=t^{-1/2}$ and $\mrd t=2\varepsilon^{-3} \mrd\varepsilon$ 
yield
\[
\mathbb{E} \left[\frac1{(\Delta_{\min}^+)^2}\right] = 2\int_0^\infty \varepsilon^{-3} \mathbb{P}[\Delta_{\min}^+<\varepsilon] \mrd\varepsilon.
\]
By Lemma~\ref{lem:minspacing}, we have $\mathbb{P}[\Delta_{\min}^+<\varepsilon] \gtrsim m^2\varepsilon$
for sufficiently small \(\varepsilon\) satisfying $\varepsilon \le c / m^2 $.
Therefore,
\[
 \mathbb{E} \left[\frac{1}{(\Delta_{\min}^+)^2}\right] \gtrsim m^2 \int_0^{c/m^2}
\varepsilon^{-2}\mrd\varepsilon = +\infty.
\]
Consequently,
\[
 \mathbb{E}\left[\frac1{\sigma_{r}^2(\pA)}\right] = +\infty.
\]
\end{proof}

\subsection{Proof of Theorem \ref{Thm2}}
\begin{proof}[Proof of Theorem \ref{Thm2}]
Similar to the proof Theorem \ref{Thm:Divergence}, we have  the following lemma:
\begin{lemma}
\label{lem:minspacing2}
$t_l=\frac{l-1}{m}+\delta_l, l=1,\dots,m$, where $\delta_l\overset{i.i.d.}{\sim} \Unif(0,\frac{1}{m})$, and assume that  $\varepsilon \le O(1 / m^{3/2})$. 
Then there exist constants $c,C>0$ such that for all sufficiently small  $\varepsilon>0 $,
\[
c m^3 \varepsilon^2 \leq \mathbb{P}[ \Delta^+_{\min}\leq \varepsilon] \leq C m^3 \varepsilon^2,
\]
\end{lemma}
where $\Delta^+_{\min}$ is defined in \eqref{def:delta}.
\begin{proof}[Proof of Lemma \ref{lem:minspacing2}]
For any fixed $i$ with pair $(t_i,t_{i+1})$ and $\epsilon \in (0,1/m)$, we can obtain formula 
\[
 \mathbb{P}[ t_{i+1} - t_i \leq \varepsilon]=\frac{(\varepsilon -\frac{1}{m}+T)^2}{2 T^2} = \frac{m^2 \varepsilon^2}{2}
\]
by setting $T= 1/m$. 
Recalling the definition of $\Delta_{\min}$ in Proposition \ref{pro:sigmin},
by a union bound,
\begin{align*}
 \mathbb{P}[\Delta_{\min}\le\varepsilon] &=\mathbb{P}\left[\bigcup_{i=1}^{m-1}(t_{i+1} - t_i\leq \varepsilon) \right] \\
&\leq \sum_{i=1}^{m-1} \mathbb{P}[t_{i+1}-t_i\le\varepsilon]
\le 2(m-1) \cdot \frac{m^2 \varepsilon^2}{2} \le m^3 \varepsilon^2.
\end{align*}
A matching lower bound follows by considering disjoint events
\[
E_k = \{ t_{2k}-t_{2k-1} \le\varepsilon \}, \quad k=1,\dots, \left\lfloor\frac{m}{2}\right\rfloor.
\]
The events $\{E_k\}_k$  are mutually independent because each event involves only disjoint random variables. It implies that 
$\mathbb{P}[E_k] = m^2 \varepsilon^2 /2$. 
Since $\bigcup_{k=1}^{\lfloor m/2\rfloor} E_k \subset \{ \Delta_{\min} \leq \varepsilon\}$, 
we then get
\begin{align*}
\mathbb{P}[\Delta_{\min} \leq \varepsilon] \geq \mathbb{P}\left[ \bigcup_{k=1}^{\lfloor m/2\rfloor} E_k \right] = 1 - \Big(1- \frac{m^2 \varepsilon^2}{2}\Big)^{\left\lfloor m /2 \right\rfloor}.
\end{align*}
When $\varepsilon \le O(\frac{1}{m^{3/2}})$, we have
$\mathbb{P}[\Delta_{\min} \leq \varepsilon] \geq m^3 \varepsilon^2 /8$. 

Since $t_1,\dots,t_m$ are continuously distributed, we have  $\mathbb P[t_i=t_j]=0$, for $i\neq j$. Hence, $\Delta_{\min}=\Delta_{\min}^+$ holds almost surely. 
Therefore, the same estimate holds for $\Delta_{\min}^+$.
\end{proof}
Then, based on Lemme \ref{lem:minspacing2} and Proposition \ref{pro:sigmin}, we have
\begin{align*}
\mathbb{E}\left[\frac{1}{\sigma_{r}^2(\pA)}\right] &\gtrsim  \mathbb{E} \left[ \frac{1}{(\Delta^+_{\min})^2} \right]
\gtrsim  \int_{0}^{\eta} \frac{1}{\varepsilon^2} \mrd \mathbb{P}[\Delta^+_{\min} \leq \varepsilon] \\
&\gtrsim  \int_{0}^{\eta} \frac{m^3\varepsilon}{\varepsilon^2} \mrd \varepsilon \gtrsim m^3 \int_{0}^{\eta} \frac{1}{\varepsilon} \mrd \varepsilon = + \infty, 
\end{align*}
where $\eta$ is a sufficiently small positive constant.

The asymptotic regimes can be obtained similarly.
\end{proof}

\subsection{Proof of Theorem \ref{Equidistant}}
\begin{proof}[Proof of Theorem \ref{Equidistant}]
For $1\le l,l'\le m$,
\[ 
(\pA\pA^*)_{l,l'} = \sum_{k\in\Gamma} e^{2\pi \mri k(t_l-t_{l'})}.
\vspace{-0.2cm}
\]
Since $t_l-t_{l'} = \frac{l-l'}{m}$, the matrix $\pA\pA^*$ is circulant.

The eigenvectors of a circulant matrix are the discrete Fourier modes. Therefore, for $s=0,\dots,m-1$, the corresponding eigenvalues are
\[
 \lambda_s = \sum_{r=0}^{m-1} (\pA\pA^*)_{1,r+1} e^{-2\pi \mri sr/m}.
\]
Substituting the first row of $\pA\pA^*$ yields
\[
 \lambda_s = \sum_{k\in\Gamma} \sum_{r=0}^{m-1} e^{2\pi \mri (k-s)r/m}.
\]
Using
\[
\sum_{r=0}^{m-1} e^{2\pi \mri (k-s)r/m} = 
\begin{cases}
m, & k\equiv s \pmod m,
\\[4pt]
0, & k\not\equiv s \pmod m,
\end{cases}
\]
we can obtain
\[
 \lambda_s = m \cdot \#\{ k \in \Gamma : k \equiv s \ (\bmod\ m) \}, \quad s=0,\dots,m-1.
\]
Write $N = am + b$ with $a = \lfloor N/m \rfloor$, $0 \le b < m$.
Since $\Gamma$ consists of $N$ consecutive integers, exactly $b$ residue classes modulo $m$ occur $a+1$ times, while the remaining $m-b$ residue classes occur $a$ times.

If $m\ge N$, then $a=0$ and $b=N$. In this case, $N$ residue classes occur once and the remaining $m-N$ residue classes do not occur. Hence
\vspace{-0.2cm}
\[
\lambda_s=
\begin{cases}
m, & \text{for } N \text{ residue classes}, \\[4pt]
0, & \text{for } m-N \text{ residue classes}.
\end{cases}
\vspace{-0.2cm}
\]
Since $r=N$, every nonzero eigenvalue equals $m$.

If $m<N$, when $b=0$, then every residue class occurs exactly $a$ times. Hence
\[
 \lambda_s=ma, \qquad s=0,\dots,m-1.
\]
Therefore, every nonzero eigenvalue equals $ma$.
Since $a = \left\lceil N /m \right\rceil = \left\lfloor N/m \right\rfloor$, we can also say 
the largest eigenvalue equals $ma = m\left\lceil N/m \right\rceil$, 
and the smallest positive eigenvalue equals $ma = m\left\lfloor N/m \right\rfloor$. 
Assume now that $b>0$. Then
\vspace{-0.2cm}
\[
\lambda_s =
\begin{cases}
m(a+1), & \text{for } b \text{ residue classes}, \\[4pt]
ma, & \text{for } m-b \text{ residue classes}.
\end{cases}
\vspace{-0.2cm}
\]
Consequently, the largest eigenvalue equals $m(a+1) = m\left\lceil N/m \right\rceil$, and the smallest positive eigenvalue equals $ma = m\left\lfloor N/m \right\rfloor$.

Since $\pA\pA^*$ is Hermitian positive semidefinite,
its nonzero singular values coincide with its nonzero eigenvalues.
This completes the proof.
\end{proof}

\subsection{Proof of Theorem \ref{thm:random_smallT_explicit}}
\begin{proof}[Proof of Theorem \ref{thm:random_smallT_explicit}]

The proof combines the perturbation estimate established above with the explicit spectrum of the equidistant matrix $\tilde{\pA}\tilde{\pA}^*$. 
Let $\theta_1\ge \theta_2\ge \cdots \ge \theta_r>0$ denote the nonzero singular values of $\tilde{\pA}\tilde{\pA}^*$.
By Weyl's inequality, we know
\[
|\sigma_i(\pA\pA^*)-\theta_i| \le \|\pP\|, \qquad i=1,\ldots,r.
\]
Since $\eta<\theta_r$ and $\|\pP\|\le \eta$, we obtain $\theta_i-\|\pP\|>0, i=1,\ldots,r$. 
Hence, 
\[
 \frac{1}{\theta_i+\|\pP\|}
\le
\frac{1}{\sigma_i(\pA\pA^*)}
\le
\frac{1}{\theta_i-\|\pP\|}.
\]
Summing over $i$ gives
\[
\sum_{i=1}^{r} \frac{1}{\theta_i+\|\pP\|} \le \sum_{i=1}^{r} \frac{1}{\sigma_i(\pA\pA^*)} \le \sum_{i=1}^{r}\frac{1}{\theta_i-\|\pP\|}.
\]
Taking expectations yields
\[
 \mathbb{E} \Bigg[ \sum_{i=1}^{r} \frac{1}{\sigma_i(\pA\pA^*)} \Bigg]
\ge
\sum_{i=1}^{r} \mathbb{E} \Bigg[\frac{1}{\theta_i+\|\pP\|} \Bigg].
\]
Since the function $f_i(x)=\frac1{\theta_i+x}$ is convex on \([0,\infty)\), Jensen's inequality implies
\[
\mathbb{E} \Bigg[ \frac{1}{\theta_i+\|\pP\|} \Bigg]
\ge \frac{1}{\theta_i+\mathbb{E} \|\pP\|}.
\]
Therefore
\[
\mathbb{E} \Bigg[\sum_{i=1}^{r} \frac{1}{\sigma_i(\pA\pA^*)} \Bigg]
\ge \sum_{i=1}^{r} \frac{1}{ \theta_i+\mathbb{E}\|\pP\|
}.
\]
Using $\mathbb{E}\|\pP\| \le 2\alpha\sqrt{mS_1}+mS_1$, we obtain
\[
\mathbb{E} \Bigg[ \sum_{i=1}^{r} \frac{1}{\sigma_i(\pA\pA^*)} \Bigg]
\ge \sum_{i=1}^{r} \frac{1}{\theta_i+2\alpha\sqrt{mS_1}+mS_1}.
\]

For the upper bound, splitting according to the events
\[
\{\|\pP\|\le t_0\} \qquad\text{and}\qquad \{\|\pP\|>t_0\},
\]
gives
\[
 \mathbb{E} \Bigg[\frac{1}{\theta_i-\|\pP\|}\Bigg]
\le \frac{1}{\theta_i-t_0} +\mathbb{P}[\|\pP\|>t_0] \left(\frac{1}{\theta_i-\eta}-\frac{1}{\theta_i-t_0}\right).
\]
Summing over $i$ yields
\[
 \mathbb{E}\Bigg[\sum_{i=1}^{r}\frac{1}{\sigma_i(\pA\pA^*)}\Bigg]
\le\sum_{i=1}^{r}\frac{1}{\theta_i-t_0}+\mathbb{P}[\|\pP\|>t_0]\sum_{i=1}^{r}\left(\frac{1}{\theta_i-\eta}-\frac{1}{\theta_i-t_0}\right).
\]
By the Bernstein estimate established above,
\[
\mathbb{P}[\|\pP\|>t_0] \le p_{t_0}.
\]
Hence, we have
\[
 \mathbb{E} \Bigg[ \sum_{i=1}^{r} \frac{1}{\sigma_i(\pA\pA^*)} \Bigg]
\le \sum_{i=1}^{r} \frac{1}{\theta_i-t_0} + p_{t_0} \sum_{i=1}^{r} \left( \frac{1}{\theta_i-\eta} - \frac{1}{\theta_i-t_0} \right).
\]
\smallskip

\noindent
{\bf Case \(m\ge N\).}
By Theorem~\ref{Equidistant}, we know that $\theta_1=\cdots=\theta_N=m$. 
Moreover,
\[
\alpha=\|\tilde{\pA}\|=\sqrt m.
\]
And if  $T < \frac{\sqrt{2}-1}{\pi N^{3/2}}$, 
we have $\eta<\theta_r$.
Substituting these identities into the general bounds gives
\[
 \mathbb{E}\Bigg[\sum_{i=1}^{r}\frac{1}{\sigma_i(\pA\pA^*)}\Bigg] \ge \frac{N}{m+2m\sqrt{S_1}+mS_1} = \frac{N}{m(1+\sqrt{S_1})^2},
\]
and
\[
 \mathbb{E} \Bigg[ \sum_{i=1}^{r} \frac{1}{\sigma_i(\pA\pA^*)} \Bigg] \le N \left( \frac{1}{m-t_0}+p_{t_0}\Bigl(\frac{1}{m-\eta}-\frac{1}{m-t_0}
\Bigr)\right).
\]
This proves part {\rm(i)}.

\noindent
{\bf Case \(m<N\).}
Write $N=am+b$, $0\le b<m$. 
Again by Theorem~\ref{Equidistant}, the nonzero eigenvalues of $\tilde{\pA}\tilde{\pA}^*$ are
\[
\theta_1=\cdots=\theta_b=m(a+1), \quad \text{and} \quad \theta_{b+1} = \cdots = \theta_m = ma.
\]
Furthermore,
\[
\alpha=\|\tilde{\pA}\| = \sqrt{m(a+1)}.
\]
And if  $T < \frac{\sqrt{2a+1} - \sqrt{a+1}}{\pi N^{3/2}}$, 
we have $\eta<\theta_r$.
Substituting these values into the lower bound gives
\begin{align*}
&\mathbb{E} \left[\sum_{i=1}^{r} \frac{1}{\sigma_i(\pA\pA^*)}\right] \ge \\
&\tiny \qquad \frac{m-b}{m\Bigl(a+2\sqrt{(a+1)S_1}+S_1\Bigr)} + \frac{b}{m\Bigl(a+1+2\sqrt{(a+1)S_1}+S_1\Bigr)} 
\end{align*}

Similarly, substituting into the upper bound yields
\begin{align*}
 \mathbb{E} \left[\sum_{i=1}^{r} \frac{1}{\sigma_i(\pA\pA^*)}\right]
\le
(m-b)\left( \frac{1}{ma-t_0}+ p_{t_0}\Bigl(\frac{1}{ma-\eta}-\frac{1}{ma-t_0}\Bigr) \right) \\
\tiny +
b\left(\frac{1}{m(a+1)-t_0}+p_{t_0}\Bigl(\frac{1}{m(a+1)-\eta}-\frac{1}{m(a+1)-t_0}\Bigr) \right).
\end{align*}
This proves part {\rm(ii)}.
\end{proof}

\section{Proof of Corollaries and Propositions}\label{sec:4}
In this section, we provide  proofs of the corollaries and propositions.
\subsection{Proof of Proposition \ref{pro4}}
\begin{proof}[Proof of Proposition \ref{pro4}]
We only consider the case $m \geq N$, and the case $m<N$ is similar. \\
\indent For any fixed choice of $N$ rows, the corresponding $N\times N$ subdeterminant of $\pA$ is a trigonometric polynomial in $(t_1,\dots , t_m)$.
Moreover, this polynomial is not identically zero, because if the selected nodes are pairwise distinct, the corresponding submatrix is a classical Vandermonde matrix and hence has nonvanishing determinant.
Since the zero set of a nontrivial trigonometric polynomial has Lebesgue measure zero, and since the distributions of $t_1, \dots, t_m$ are absolutely continuous with respect to the Lebesgue measure, each such determinant is nonzero almost surely.
Therefore the event that all $N\times N$ minors vanish simultaneously has probability zero. Hence, $\mathbb{P}[\rank(\pA)=N]=1$. 
Thus $\pA^*\pA$ is almost surely Hermitian positive definite. Its eigenvalues satisfy 
\[
0 < \lambda_{r}(\pA^*\pA) \leq \dots \leq \lambda_{1}(\pA^*\pA) < \infty \quad \text{almost surely}.
\]
Since $\pA^*\pA$ is finite-dimensional, $\sum_{i=1}^{r} \frac{1}{\sigma_i^2(\pA)}$ is finite almost surely.
\end{proof}

\subsection{Proof Sketch of Proposition \ref{pro3}}
\begin{proof}[Proof Sketch of Proposition \ref{pro3}]
We set $\pv(t) = \left( e^{2\pi  \mathrm{i}kt} \right)_{k\in \Gamma} \in \mathbb{C}^N$ and 
$\pX =  \pv(t_l)\pv(t_l)^* - \pI_N$.  
Since $t_l$ is chosen independently and uniformly at random, $\pv(t_l)$ and $X_l$ are independent for different $l$. We have
\[
\mathbb{E} \left[ \pv(t_l)\pv(t_l)^* \right]_{k_1,k_2} = \mathbb{E}\left[e^{2\pi  \mathrm{i}(k_1-k_2)t} \right] = \left\{
\begin{array}{lr}
1, &k_1=k_2,\\
0, &k_1 \ne k_2.
\end{array}
\right.
\]
We have
\[
\| \pX_l\| \leq \| \pv(t_l)\pv(t_l)^* \| + \|\pI_N\| = N+1 \triangleq  R
\]
and $\sigma^2 \triangleq \left\| \sum_{l=1}^m \mbE \pX_l^2 \right\|$. 
Since $
\pX_l^2 = [\pv(t_l)\pv(t_l)^* - \pI_N]^2 \leq 2\pv(t_l)\pv(t_l)^*$
and $\mbE \pX_l^2 \preceq 2N \pI_N$,
we have $\sigma^2 \leq 2mN$.
\end{proof}

\subsection{Proof of Proposition \ref{pro:sigmin}}
\begin{proof}[Proof of Proposition \ref{pro:sigmin}]
For pairwise distinct nodes $\{t_l\}_{l=1}^{m}$, 
denote $(\bar{i},\bar{j}) = \arg\min_{i\ne j} \dist_T(t_{i},t_j)$. 
Set $\pa(t_i) = \left( e^{2\pi \mathrm{i} kt_i} \right)_{k=-q}^{q} \in \mathbb{C}^N$ as the $i$-row of $\pA$.
Consider the vector $\pz=\pe_{\bar{i}}-\pe_{\bar{j}}\in\mathbb C^m$, for which $\|\pz\|_2=\sqrt{2}$.
Then, we have $\pA\pz = \pa(t_{\bar{i}})-\pa(t_{\bar{j}})$. 
We compute
\begin{align*}
&\|  \pa(t_{\bar{i}})-\pa(t_{\bar{j}}) \|_2^2 = \sum_{k=-q}^q |e^{2\pi \mri k \Delta_{\min}}-1|^2 \\
\le&
(2\pi)^2 \Delta_{\min}^2\sum_{k=-q}^q k^2 
= \frac{4\pi^2q(q+1)(2q+1)}{3} \Delta_{\min}^2
\lesssim N^3\Delta_{\min}^2,
\end{align*}
where the first inequality comes from $|e^{\mri x}-1|\leq |x|$. By the variational characterization of the smallest singular value, we have
\[
\sigma_{r}(\pA) = \min_{\|\px\|_2=1} \|\pA\px\|_2
\le \frac{\|\pA\pz\|_2}{\|\pz\|_2}  \lesssim N^{\frac{3}{2}} \Delta_{\min}.
\]
For any nodes $\{t_l\}_{l=1}^{m} \subset [0,1)^m$, the argument follows in exactly the same manner 
by removing coincident nodes $t_l$.
\end{proof}

\subsection{Proof of Proposition \ref{prop3_8}}
\begin{proof}[Proof of Proposition \ref{prop3_8}]
Set $X =  \delta_2 -  \delta_1$. It is easy to get its Probability Density Function (PDF): 
\begin{align}
f_X(x) = \left\{
\begin{array}{lr}
\frac{T - |x|}{T^2}, &y \in (-T, T),\\
0, &\text{otherwise}.
\end{array}
\right.
\end{align}
This is a symmetrically triangular distribution, supported by $(-T, T)$, with a peak at $x = 0$ and a peak value of $\frac{1}{T}$. 
Since $Y = c + \delta_2 -  \delta_1 \geq 0$, we can obtain the PDF of $Y$
\begin{align} \label{PDF_Y}
f_Y(y) = f_X\Big(y-\frac{1}{m}\Big) = \left\{
\begin{array}{lr}
\frac{T - |y-c|}{T^2}, &y \in (c -T, c+T),\\
0, &\text{otherwise}.
\end{array}
\right.
\end{align}
The support interval is $(c-T, c+T)$. The PDF is a triangular distribution symmetric about $y=c$, with a peak value of $1/T$ at $y=c$. Based on \eqref{PDF_Y}, we can get 
\begin{itemize}
\item When $y < c-T$, 
\[
F_Y(y) = 0;
\]
\item When $c-T \leq y \leq c$, 
\[
F_Y(y) = \int_{c-T}^{y} \frac{T-(c-t)}{T^2} \mrd t = \frac{(y-c+T)^2}{2 T^2};
\]
\item When $c < y \leq c+T$, 
\[
F_Y(y) = 1- \int^{c+T}_{y} \frac{T-(t-c)}{T^2} \mrd t = 1 - \frac{(y-c-T)^2}{2 T^2};
\]
\item When $y > c+T$, 
\[
F_Y(y) = 1.
\]
\end{itemize}
Setting $c = \frac{1}{m}$, we can complete the proof.
\end{proof}

\subsection{Proof of Corollary \ref{cor3_11}}
\begin{proof}[Proof of Corollary \ref{cor3_11}]
If $m\ge N$, then all $r=N$ nonzero singular values are equal to $m$.

If $m<N$, then $r=m$, and by Theorem \ref{Equidistant}, there are exactly $b$ singular values equal to $m(a+1)$ and $m-b$ singular values equal to $ma$. Therefore, we have
\vspace{-0.2cm}
\[
 \sum_{i=1}^{r}\frac1{\sigma_i(\pA\pA^*)} = \frac{b}{m(a+1)} + \frac{m-b}{ma} = \frac{m+N-2b}{ma(a+1)}.
\]
\end{proof}

\subsection{Proof of Corollary \ref{cor3_12}}
\begin{proof}[Proof of Corollary \ref{cor3_12}]
If $m<N$, then $r=m$, and by Theorem \ref{Equidistant}, there are exactly $b$ singular values equal to $m(a+1)$ and $m-b$ singular values equal to $ma$. Therefore, we have
\vspace{-0.2cm}
\[
ma \le \sigma_i(\pA\pA^*) \le m(a+1) \quad \text{and}\quad \frac{1}{m(a+1)} \le \frac{1}{\sigma_i(\pA\pA^*)} \le \frac{1}{ma}.
\vspace{-0.2cm}
\]
Then we can get 
\vspace{-0.2cm}
\[
 \frac{1}{\Bigl\lfloor\frac{N}{m}\Bigr\rfloor+1}=\frac{1}{a+1} \le \sum_{i=1}^{r} \frac{1}{\sigma_i(\pA\pA^*)} \le \frac{1}{a}= \frac{1}{\Bigl\lfloor\frac{N}{m}\Bigr\rfloor}.
\vspace{-0.2cm}
\]
\end{proof}

\end{appendices}


\bmhead{Funding}
This work was  supported in part by the NSFC under grant numbers 12371101, 12401124 and the Natural Science Foundation of Zhejiang Province under grant numbers LQN25A010002.

\bibliography{reference}

@article{farrell2011limiting,
  title={Limiting empirical singular value distribution of restrictions of discrete Fourier transform matrices},
  author={Farrell, Brendan},
  journal={Journal of Fourier Analysis and Applications},
  volume={17},
  number={4},
  pages={733--753},
  year={2011},
  publisher={Springer}
}

@article{potts2003fast,
  title={Fast summation at nonequispaced knots by NFFT},
  author={Potts, Daniel and Steidl, Gabriele},
  journal={SIAM Journal on Scientific Computing},
  volume={24},
  number={6},
  pages={2013--2037},
  year={2003},
  publisher={SIAM}
}

@book{christensen2003introduction,
  title={An introduction to frames and Riesz bases},
  author={Christensen, Ole and others},
  volume={7},
  year={2003},
  publisher={Springer},
  address={Boston: Birkh{\"a}user}
}

@article{kunis2020smallest,
  title={On the smallest singular value of multivariate Vandermonde matrices with clustered nodes},
  author={Kunis, Stefan and Nagel, Dominik},
  journal={Linear Algebra and its Applications},
  volume={604},
  pages={1--20},
  year={2020},
  publisher={Elsevier}
}

@article{barnett2022exponentially,
  title={How exponentially ill-conditioned are contiguous submatrices of the Fourier matrix?},
  author={Barnett, Alex H},
  journal={Siam Review},
  volume={64},
  number={1},
  pages={105--131},
  year={2022},
  publisher={SIAM}
}

@article{xie2022overparameterization,
  title={Overparameterization and generalization error: weighted trigonometric interpolation},
  author={Xie, Yuege and Chou, Hung-Hsu and Rauhut, Holger and Ward, Rachel},
  journal={SIAM Journal on Mathematics of Data Science},
  volume={4},
  number={2},
  pages={885--908},
  year={2022},
  publisher={SIAM}
}

@article{chen2024conditioning,
  title={Conditioning of random Fourier feature matrices: double descent and generalization error},
  author={Chen, Zhijun and Schaeffer, Hayden},
  journal={Information and Inference: A Journal of the IMA},
  volume={13},
  number={2},
  pages={iaad054},
  year={2024},
  publisher={Oxford University Press}
}

@article{belkin2020two,
  title={Two models of double descent for weak features},
  author={Belkin, Mikhail and Hsu, Daniel and Xu, Ji},
  journal={SIAM Journal on Mathematics of Data Science},
  volume={2},
  number={4},
  pages={1167--1180},
  year={2020},
  publisher={SIAM}
}

@article{muthukumar2020harmless,
  title={Harmless interpolation of noisy data in regression},
  author={Muthukumar, Vidya and Vodrahalli, Kailas and Subramanian, Vignesh and Sahai, Anant},
  journal={IEEE Journal on Selected Areas in Information Theory},
  volume={1},
  number={1},
  pages={67--83},
  year={2020},
  publisher={IEEE}
}

@article{mei2022generalization,
  title={The generalization error of random features regression: Precise asymptotics and the double descent curve},
  author={Mei, Song and Montanari, Andrea},
  journal={Communications on Pure and Applied Mathematics},
  volume={75},
  number={4},
  pages={667--766},
  year={2022},
  publisher={Wiley Online Library}
}

@book{grochenig2001foundations,
  title={Foundations of time-frequency analysis},
  author={Gr{\"o}chenig, Karlheinz},
  volume={359},
  year={2001},
  publisher={Springer},
  address={Boston: Birkh{\"a}user}
}

@article{feichtinger1992irregular,
  title={Irregular sampling theorems and series expansions of band-limited functions},
  author={Feichtinger, Hans G and Gr{\"o}chenig, Karlheinz},
  journal={Journal of Mathematical Analysis and Applications},
  volume={167},
  number={2},
  pages={530--556},
  year={1992},
  publisher={Elsevier}
}

@inproceedings{kadets1964exact,
  title={The exact value of the Paley--Wiener constant},
  author={Kadets, Mikhail Iosifovich},
  booktitle={Doklady Akademii Nauk},
  volume={155},
  number={6},
  pages={1253--1254},
  year={1964},
  organization={Russian Academy of Sciences}
}

@article{vershynin2019high,
  title={High-dimensional probability},
  author={Vershynin, Roman},
  journal={Cambridge Series in Statistical and Probabilistic Mathematics},
  volume={47},
  year={2019}
}

@article{tao2010random,
  title={Random matrices: The distribution of the smallest singular values},
  author={Tao, Terence and Vu, Van},
  journal={Geometric And Functional Analysis},
  volume={20},
  number={1},
  pages={260--297},
  year={2010},
  publisher={Springer}
}

@article{rudelson2008littlewood,
  title={The Littlewood--Offord problem and invertibility of random matrices},
  author={Rudelson, Mark and Vershynin, Roman},
  journal={Advances in Mathematics},
  volume={218},
  number={2},
  pages={600--633},
  year={2008},
  publisher={Elsevier}
}

@article{edelman1988eigenvalues,
  title={Eigenvalues and condition numbers of random matrices},
  author={Edelman, Alan},
  journal={SIAM Journal on Matrix Analysis and Applications},
  volume={9},
  number={4},
  pages={543--560},
  year={1988},
  publisher={SIAM}
}

@article{adcock2012generalized,
  title={A generalized sampling theorem for stable reconstructions in arbitrary bases},
  author={Adcock, Ben and Hansen, Anders C},
  journal={Journal of Fourier Analysis and Applications},
  volume={18},
  number={4},
  pages={685--716},
  year={2012},
  publisher={Springer}
}

@article{rauhut2007random,
  title={Random sampling of sparse trigonometric polynomials},
  author={Rauhut, Holger},
  journal={Applied and Computational Harmonic Analysis},
  volume={22},
  number={1},
  pages={16--42},
  year={2007},
  publisher={Elsevier}
}

@article{candes2006robust,
  title={Robust uncertainty principles: Exact signal reconstruction from highly incomplete frequency information},
  author={Cand{\`e}s, Emmanuel J and Romberg, Justin and Tao, Terence},
  journal={IEEE Transactions on Information Theory},
  volume={52},
  number={2},
  pages={489--509},
  year={2006},
  publisher={IEEE}
}

@article{averbuch2001fast,
  title={Fast slant stack: A notion of radon transform for data in a cartesian grid which is rapidly computible, algebraically exact, geometrically faithful and invertible},
  author={Averbuch, Amir and Coifman, RR and Donoho, DL and Israeli, Moshe and Walden, Johan},
  journal={SIAM Scientific Computing},
  volume={37},
  number={3},
  pages={192--206},
  year={2001},
  publisher={to appear}
}

@article{dutt1993fast,
  title={Fast Fourier transforms for nonequispaced data},
  author={Dutt, Alok and Rokhlin, Vladimir},
  journal={SIAM Journal on Scientific Computing},
  volume={14},
  number={6},
  pages={1368--1393},
  year={1993},
  publisher={SIAM}
}

@article{strohmer1997computationally,
  title={Computationally attractive reconstruction of bandlimited images from irregular samples},
  author={Strohmer, Thomas},
  journal={IEEE Transactions on Image Processing},
  volume={6},
  number={4},
  pages={540--548},
  year={1997},
  publisher={IEEE}
}

@article{rauth1998smooth,
  title={Smooth approximation of potential fields from noisy scattered data},
  author={Rauth, Michael and Strohmer, Thomas},
  journal={Geophysics},
  volume={63},
  number={1},
  pages={85--94},
  year={1998},
  publisher={Society of Exploration Geophysicists}
}

@inproceedings{strohmer1996recover,
  title={How to recover smooth object boundaries in noisy medical images},
  author={Strohmer, Thomas and Binder, Thomas and Sussner, M},
  booktitle={Proceedings of 3rd IEEE International Conference on Image Processing},
  volume={1},
  pages={331--334},
  year={1996},
  organization={}
}

@article{potts2001fast,
  title={Fast Fourier transforms for nonequispaced data: A tutorial},
  author={Potts, Daniel and Steidl, Gabriele and Tasche, Manfred},
  journal={Modern Sampling Theory: Mathematics and Applications},
  pages={247--270},
  year={2001},
  publisher={Springer}
}

@article{grochenig1993discrete,
  title={A discrete theory of irregular sampling},
  author={Gr{\"o}chenig, Karlheinz},
  journal={Linear Algebra and Its Applications},
  volume={193},
  pages={129--150},
  year={1993},
  publisher={Elsevier}
}

@article{bass2005random,
  title={Random sampling of multivariate trigonometric polynomials},
  author={Bass, Richard F and Gr{\"o}chenig, Karlheinz},
  journal={SIAM Journal on Mathematical Analysis},
  volume={36},
  number={3},
  pages={773--795},
  year={2005},
  publisher={SIAM}
}

@article{tropp2012user,
  title={User-friendly tail bounds for sums of random matrices},
  author={Tropp, Joel A},
  journal={Foundations of Computational Mathematics},
  volume={12},
  number={4},
  pages={389--434},
  year={2012},
  publisher={Springer}
}

@inproceedings{moitra2015super,
  title={Super-resolution, extremal functions and the condition number of Vandermonde matrices},
  author={Moitra, Ankur},
  booktitle={Proceedings of the forty-seventh annual ACM Symposium on Theory of Computing},
  pages={821--830},
  year={2015}
}

@article{candes2014towards,
  title={Towards a mathematical theory of super-resolution},
  author={Cand{\`e}s, Emmanuel J and Fernandez-Granda, Carlos},
  journal={Communications on Pure and Applied Mathematics},
  volume={67},
  number={6},
  pages={906--956},
  year={2014},
  publisher={Wiley Online Library}
}

\end{document}